\documentclass[a4paper,USenglish,cleveref,autoref,thm-restate]{lipics-v2021}

\usepackage{commands-tam}
\usepackage[utf8]{inputenc}
\usepackage{subcaption}
\usepackage{siunitx}
\usepackage{blkarray}
\usepackage{cite}

\usepackage{cleveref}

\usepackage{tikz}
\usetikzlibrary{shapes.geometric, arrows, decorations.pathmorphing, calc, decorations.pathreplacing, math}
\usetikzlibrary{tikzmark}
\tikzstyle{startstop} = [rectangle, rounded corners, 
minimum width=3cm, 
minimum height=1cm,
text centered, 
draw=black, 
fill=red!30]

\tikzstyle{io} = [trapezium, 
trapezium stretches=true, 
trapezium left angle=70, 
trapezium right angle=110, 
minimum width=3cm, 
minimum height=1cm, text centered, 
draw=black, fill=blue!30]

\tikzstyle{process} = [rectangle, 
minimum width=3cm, 
minimum height=1cm, 
text centered, 
text width=3cm, 
draw=black, 
fill=orange!30]

\tikzstyle{decision} = [diamond, 
minimum width=3cm, 
minimum height=1cm, 
text centered, 
draw=black, 
fill=green!30]
\tikzstyle{arrow} = [thick,->,>=stealth]

\newcommand{\shapeasm}[1]{\mathcal{S}_{\Box}[\mathcal{#1}]}

\newcommand{\calS}{\mathcal{S}}

\newcommand{\Z}{\mathbb{Z}}
\usepackage{stmaryrd}
\newcommand{\intinterval}[1]{\llbracket #1 \rrbracket}

\usepackage{todonotes}

\newif\ifhighlight
\highlighttrue

\newif\ifshowtext
\showtextfalse

\def\mytitle{Powers and Limitations of Synchronous Self-Assembly}
\title{\mytitle}
\titlerunning{\mytitle}
\subtitle{Non-cooperative Assemblies and Limited Synchronization}

\author{Florent Becker}{Laboratoire d’Informatique Fondamentale d’Orléans (UR4022), Université d’Orléans and INSA Centre-Val-de-Loire, France}{florent.becker@univ-orleans.fr}{}{Research made at a public, state-funded university, thanks to the guarantee of full academic freedom and perennial research funding.}

\author{Phillip Drake}{University of Arkansas, USA}{padrake@uark.edu}{}{This author's work was supported in part by NSF grant 2329908}

\author{Matthew J. Patitz}{University of Arkansas, USA}{patitz@uark.edu}{https://orcid.org/0000-0001-9287-4028}{This author's work was supported in part by NSF grant 2329908 }

\author{Ryder Smith}{Hass Hall Academy, Bentonville, AR, USA}{ryder.smithrcs@gmail.com}{}{}

\authorrunning{F.~Becker, P.~Drake, M.\,J. Patitz, R.~Smith}

\Copyright{Florent Becker, Phillip Drake, Matthew J. Patitz, and Ryder Smith}



\category{} 

\relatedversion{} 

\supplement{}

\usepackage[utf8]{inputenc}
\usepackage{subcaption}

\EventEditors{Dominic Scalise and Robert Schweller}
\EventNoEds{2}
\EventLongTitle{32nd International Conference on DNA Computing and Molecular Programming (DNA 32)}
\EventShortTitle{DNA 32}
\EventAcronym{DNA}
\EventYear{32}
\EventDate{August 3--7, 2026}
\EventLocation{Fayetteville, AR, USA}
\EventLogo{}
\SeriesVolume{}
\ArticleNo{}

\begin{document}

\maketitle

\begin{abstract}
In abstract models of algorithmic self-assembly, synchronization between attachments has emerged as a crucial distinction between the classical asynchronous model (aTAM) and a new synchronous model, the syncTAM. This paper presents recent advances in gauging the additional power afforded by the syncTAM. While it is known that the syncTAM and the aTAM are each unable to fully simulate the other, this paper offers evidence that the syncTAM is computationally significantly more powerful than the aTAM, especially in the non-cooperative setting.

The additional power of the non-cooperative syncTAM is witnessed by the following constructions, all impossible in the non-cooperative aTAM: a flagpole, a strict self-assembly of a variant of the discrete Sierpinski triangle, and the ability to build the same assemblies (modulo scale factor) as directed aTAM systems.

The second topic is that of limited synchronization, wherein, when the number of attachments is smaller than some threshold $l$, they happen synchronously, but attachments in excess of that number must wait. In that context, the precise value of $l$ is crucial, and changes to that value prevent simulation and can change which shapes can be obtained. 

\end{abstract}

\maketitle

\clearpage
\pagenumbering{arabic}

\section{Introduction and Outline}

Crafting artifacts with intricate small-scale details has been a mark of sophistication in a wide variety of human cultures. Contemporary researchers have turned this endeavor into the field of nanotechnology, with the aim to extend the realm of controlled manipulation down to the molecular and atomic scales. As it turns out, DNA is a very handy vehicle for that, especially as its ability to hold data can be harnessed for algorithmic control of its assembly. Thus was born the field of (DNA) Algorithmic Self-Assembly.

The principle is to come up with basic building blocks which are able to selectively bind and grow patterns of interest. Theoretical models serve to predict and engineer this binding and assembly. There is a variety of them and this paper will be dealing with \emph{abstract} models of self-assembly, which tend to be easier to reason about while giving broad brush insights into the possibilities of molecular self-assembly.

Arguably, the first abstract model of algorithmic self-assembly was the abstract Tile Assembly Model (aTAM) \cite{Winf98}. Since its introduction, a wide variety of results have exhibited its powerful algorithmic potential \cite{SolWin07,RotWin00,jCCSA,jSADS,IUSA,FractalsSelfAssembleSODA}. 

One striking feature which sets it apart from more classical models of computation is its asynchronous nature. In order to claim correctness for a system in the aTAM, one has to prove that it obeys one's fancy whatever order the building blocks (known as tiles) come to interact with the growing assembly. On the other hand, there is a natural synchronous variant of the aTAM, called the syncTAM which works synchronously, and thus much closer to more classical models of computation such as cellular automata. The first investigations of the syncTAM in \cite{syncTAM} have shown that the two models are uncomparable (each can do things the other cannot). In this paper, we set out to refine that distinction.

It would seem natural that a synchronous model affords its user more control, which in turn means they need fewer resources in order to implement algorithms. In particular, the \emph{temperature} parameter (a.k.a. minimum binding threshold) is a very important source of control in both the aTAM and the syncTAM. Informally, the temperature of an aTAM system tells whether tiles can do a bit of local synchronization through cooperation when they assemble (at temperature 2 and more), or not (at temperature 1). It is well-known in the aTAM that systems with temperature 1 are less powerful than systems with temperature 2 and above. They can be distinguished by intrinsic universality, defined in \cite{IUSA}, allowing \cite{temp1notIU} to state that there are temperature 2 systems which no temperature 1 system can simulate. Moreover, in two dimensions, temperature 1 systems are limited in their computational capacities, as they must be \emph{fragile} or \emph{pumpable} \cite{Temp1PathsSTOC}. On the other hand, effectively implementing temperature 2 is a tough experimental challenge. For this reason, there have been many theoretical explorations of variations of the aTAM that exhibit at least some of the power of cooperative, temperature-2, self-assembly at temperature-1 \cite{SingleNegative,CookFuSch11,Polygons,Polyominoes,GeoTiles,SFTSAFT,jDuples}. We set out to show that the synchronicity of the syncTAM can act as a kind of substitute for temperature 2 by exploring the power of temperature 1 syncTAM systems. Namely, we show that these systems are able to simulate a flagpole, thus overcoming the simulation obstacle that was identified in \cite{temp1notIU} for the aTAM. The cooperative synchronization gadget used in this system can act as a basic brick for 2-dimensional computation, allowing the strict self-assembly of a shape close to the Sierpinski Triangle. Sierpinski Triangles have been known since \cite{TreeFractals} to belong to the class of fractals which cannot be strictly self-assembled in the aTAM \cite{FractalsSelfAssembleSODA}.
This witnesses that in addition to cooperation and thus computation, the temperature 1 syncTAM is able to \emph{detect the absence} of some part of the assembly.
If synchronization was a complete substitute for cooperation, then temperature 1 systems in the syncTAM would be able to build the same assemblies (at least modulo some scale factor) as temperature 2 systems in the aTAM.
 We show that this is the case in the absence of non-determinism: when restricted to aTAM systems which are directed, the claim holds.

How much of that power remains if, instead of full synchronicity, one assumes that tiles come all at once, \emph{up to some maximal amount} $l$? This is what we investigate next, introducing the L-syncTAM. Looking for intermediates between the syncTAM and the aTAM is a natural way to try and make sense of more plausible models of algorithmic DNA self-assembly. The hope is that any property which is robust to changes along the continuum between the syncTAM and the aTAM will hold in experiments. The investigation of the L-syncTAM for various values of $l$ reveals that altering that parameter either up or down affects the dynamics  of the model, preventing universal simulation between different values of $l$, and even affecting which shapes can be assembled.

\Cref{sec:prelims} introduces the reader to the aTAM, syncTAM and L-syncTAM models, as well as to intrinsic simulation for these models. \Cref{sec:shapes} exposes the power of the syncTAM model by way of the flagpole and Sierpinski triangle constructions. \Cref{sec:directed} examines directed systems in the syncTAM, up to the ``simulation'' at temperature 1 of directed aTAM systems. Finally, \cref{sec:LsyncTAM}
 is  the investigation of limited synchronicity through the L-syncTAM model.

\section{Preliminary Definitions and Models}
\label{sec:prelims}

In this section we define the models and terminology used throughout the paper.

\subsection{The abstract Tile-Assembly Model}\label{sec:aTAM_def}

We work within the abstract Tile-Assembly Model\cite{Winf98} in 2 and 3 dimensions. The abbreviation \emph{aTAM} refers to the 2D model. These definitions are borrowed from \cite{DDDIU} and we note that \cite{RotWin00} and \cite{jSSADST} are good introductions to the model for unfamiliar readers. 

Let $\mathbb{N}$ be the set of nonnegative integers,  for $l,u \in \Z$, let $\intinterval{l, u} = \{ k \in \Z \mid l \leq k < u\}$. For $n \in \mathbb{N}$, let $[n] = \intinterval{0,n} = \{0, 1, ..., n-2, n-1\}$.

Fix $\Sigma$ to be some alphabet with $\Sigma^*$ its finite strings. A \emph{glue} $g\in\Sigma^*\times\mathbb{N}$ consists of a finite string \emph{label} and non-negative integer \emph{strength}. There is a single glue of strength $0$, referred to as the \emph{null} glue. A \emph{tile type} is a tuple $t\in(\Sigma^*\times\mathbb{N})^{4}$, thought of as a unit square with a glue on each side. A \emph{tile set} is a finite set of tile types. We always assume a finite set of tile types, but allow an infinite number of copies of each tile type to occupy locations in the $\mathbb{Z}^2$ lattice, each called a \emph{tile}. A \textit{glue set} $G_T$ is the set of all glues associated with a tile set $T$. Given a direction $d \in \{N, E, S, W\} = \calD$ and a tile set $T$, $G_{T, d}$ is the set of all glues on the edges of tiles in direction $d$ and $g_{t, d}$ is the glue $g$ on tile $t$ in direction $d$. We write str$_t(d)$ to denote the strength of $g_{t,d}$. 

Given a tile set $T$, a \emph{configuration} is an arrangement (possibly empty) of tiles in the lattice $\mathbb{Z}^2$, i.e.\ a partial function $\alpha:\mathbb{Z}^2\dashrightarrow T$. Two adjacent tiles in a configuration \emph{interact}, or are \emph{bound} or \emph{attached}, if the glues on their abutting sides are equal (in both label and strength) and have positive strength. Each configuration $\alpha$ induces a \emph{binding graph} $B_\alpha$ whose vertices are those points occupied by tiles, with an edge of weight $s$ between two vertices if the corresponding tiles interact with strength $s$. An \emph{assembly} is a configuration whose domain (as a graph) is connected and non-empty. The \emph{shape} of $X \subset \mathbb{Z}^2$ is $\{ X + (x, y) | x, y \in \Z \}$, i.e. the set of all its translations. The shape of an assembly $\alpha$ is the shape $S_\alpha$ of its domain: two assemblies have the same shape if they have the same domain modulo translation. An assembly $\alpha$ contains a shape $S$ if there is $s \in S$ such that $s \subset \dom \alpha$.  For some $\tau\in\mathbb{Z}^+$, an assembly $\alpha$ is \emph{$\tau$-stable} if every cut of $B_\alpha$ has weight at least $\tau$, i.e.\ a $\tau$-stable assembly cannot be split into two pieces without separating bound tiles whose shared glues have cumulative strength $\tau$. Given two assemblies $\alpha,\beta$, we say $\alpha$ is a \emph{subassembly} of $\beta$ (denoted $\alpha \sqsubseteq \beta$) if $\dom \alpha \subseteq \dom \beta$ and for all $p\in \dom \alpha$, $\alpha(p)=\beta(p)$ (i.e., they have tiles of the same types in all locations of $\dom \alpha$). 

A \emph{tile-assembly system} (TAS) is a triple $\calT=(T,\sigma,\tau)$, where $T$ is a tile set, $\sigma$ is a finite $\tau$-stable assembly called the \emph{seed assembly}, and $\tau\in\mathbb{Z}^+$ is called the \emph{binding threshold} (a.k.a. \emph{temperature}).
Given a TAS $\calT=(T,\sigma,\tau)$ and two $\tau$-stable assemblies $\alpha$ and $\beta$, we say that $\alpha$ \emph{$\calT$-produces} $\beta$ \emph{in one step} (written $\alpha \to^{\calT}_1 \beta$) if $\alpha \sqsubseteq \beta$ and $|S_\beta \setminus S_\alpha| = 1$.
That is, $\alpha \to^{\calT}_1 \beta$ if $\beta$ differs from $\alpha$ by the addition of a single tile.
The \emph{$\calT$-frontier} is the set $\partial^{\calT}\alpha = \bigcup_{\alpha \to^{\calT}_1 \beta} S_\beta \setminus S_\alpha$ of locations in which a tile could $\tau$-stably attach to $\alpha$.
When $\calT$ is clear from context we simply refer to these as the \emph{frontier} locations.

We use $\mathcal{A}^T$ to denote the set of all assemblies of tiles in tile set $T$. Given a TAS $\calT=(T, \sigma, \tau)$, a sequence of $k\in\mathbb{Z}^+ \cup \{\infty\}$ assemblies $\alpha_0, \alpha_1, \ldots$ over $\mathcal{A}^T$ is called a \emph{$\calT$-assembly sequence} if, for all $1\le i < k$, $\alpha_{i-1} \to^{\calT}_1 \alpha_i$. The \emph{result} $\lim \alpha$ of an assembly sequence $\alpha$ is the unique limiting assembly of the sequence. For finite assembly sequences, this is the final assembly; whereas for infinite assembly sequences, this is the assembly consisting of all tiles from any assembly in the sequence. We say that \emph{$\alpha$ $\calT$-produces $\beta$} (denoted $\alpha\to^{\calT} \beta$) if there is a $\calT$-assembly sequence starting with $\alpha$ whose result is $\beta$. We say $\alpha$ is \emph{$\calT$-producible} if $\sigma\to^{\calT}\alpha$ and write $\prodasm{\calT}$ to denote the set of $\calT$-producible assemblies. We say $\alpha$ is \emph{$\calT$-terminal} if $\alpha$ is $\tau$-stable and there exists no assembly that is $\calT$-producible from $\alpha$. We denote the set of $\calT$-producible and $\calT$-terminal assemblies by $\termasm{\calT}$. If $|\termasm{\calT}| = 1$, i.e., there is exactly one terminal assembly, we say that $\calT$ is \emph{directed}. When $\calT$ is clear from context, we may omit $\calT$ from notation.

Given TAS $\calT$, we define $\shapeasm{T} = \{S_\alpha \mid \alpha \in \termasm{T}\}$ as the set of all shapes of terminal assemblies in $\calT$. Clearly, if $\calT$ is directed $|\shapeasm{T}| = 1$. However, even if $\calT$ is not directed it may be the case that $|\shapeasm{T}| = 1$ (see \cite{BryChiDotKarSekTOC} for an example). In such a case we call the system \emph{shape directed} and say that $\calT$ \emph{strictly self-assembles} $S_\alpha$, where $S_\alpha \in \shapeasm{T}$.

\subsection{The Synchronous Tile Assembly Model}

The Synchronous Tile Assembly Model (syncTAM) is the variant of the aTAM where at every step of assembly, rather that a single location of the frontier being non-deterministically selected to receive a tile, \emph{all} frontier locations simultaneously receive a tile each. As the frontier of a system can grow arbitrarily large, the number of tiles added per assembly step may also grow arbitrarily large. As in the aTAM, if any single frontier location allows $\tau$-strength binding of tiles of multiple types, one of those types is non-deterministically chosen for each such location.

We formally define the syncTAM in the same way as the aTAM with the exception of the definition of a single step in the assembly sequence. Given a syncTAS $\calS = (T, \sigma, \tau)$ and two $\tau$-stable assemblies $\alpha$ and $\beta$, we say that $\alpha$ \emph{$\calS$-produces $\beta$ in one step}, written $\alpha \to^{\calS}_1 \beta$, if $\alpha \sqsubseteq \beta$ and $\dom(\beta) \setminus \dom(\alpha) = \partial^{\calS}(\alpha)$. Note that like in the aTAM, we define $\partial^{\calS}(\alpha)$ to be the set of locations where a tile can $\tau$-stably attach to $\alpha$.  The syncTAM inherits all other definitions and notation from the aTAM.

To distinguish between a tile assembly system (a.k.a. TAS) in the aTAM versus one in the syncTAM, we will use the term \emph{syncTAS} to refer to a TAS in the syncTAM.





\subsection{Limited Synchronicity Tile Assembly Model}
The Limited Synchronocity Tile Assembly Model (L-syncTAM) is a variant of the syncTAM (introduced in \cite{syncTAM}) wherein instead of all frontier locations simultaneously receiving a tile each step, some predefined nonzero integer number of frontier locations each receive a tile. We define this number as the \emph{synchronocity} for a given system, and denote it as $l \in \mathbb{Z}^+$. At every assembly step, $l$ frontier locations are filled (unless the frontier is smaller than $l$, in which case all frontier locations are filled).

Formally, we define the L-syncTAM similarly to the syncTAM. We call a tile assembly system $\calL$ in the L-syncTAM an \emph{L-syncTAS} and define it as $\calL = (T, \sigma, \tau, l)$, where $l$ is its synchronicity and $T$, $\sigma$, and $\tau$ are defined as previously. As with the syncTAM, given a $\tau$-stable assembly $\alpha$ in $\calL$, we define $\partial^{\calL}(\alpha)$ to be the set of locations where a tile can $\tau$-stably attach to $\alpha$, and the L-syncTAM inherits all other definitions and notations from the aTAM, with the following refinement to the definition of an assembly step.

Given two $\tau$-stable assemblies $\alpha$ and $\beta$ in $\calL$, we say that $\alpha$ \emph{$\calL$-produces $\beta$ in one step}, written $\alpha \to_1^\calL \beta$ if $\alpha  \sqsubseteq \beta$ and one of the two following are true:

\begin{enumerate}
    \item $\dom(\beta) \setminus \dom(\alpha) = \partial^\calL(\alpha)$ and $|\partial^\calL(\alpha)| \le l$

    \item $\dom(\beta) \setminus \dom(\alpha) \subseteq \partial^\calL(\alpha)$ and $|\dom(\beta)\setminus \dom(\alpha)|=l$
\end{enumerate} 

The first case holds when the frontier is no larger than $l$. In the second case, when the frontier is larger, the particular subset of frontier locations that receive tiles at a step is non-deterministically selected from the full set of frontier locations. Note that since binding locations are selected only from the current frontier, $\partial^\calL(\alpha)$, where tiles can bind to the existing assembly $\alpha$ with strength $\ge \tau$, for a tile in location $\vec{l} \in \mathbb{Z}^2$, where $\vec{l} \not \in \dom \alpha$, to attach to $\alpha$ it must form one or more bonds with strengths totaling $\ge \tau$ with one or more tiles already in $\alpha$, regardless of how many bonds it may form with other tiles attaching during the same step.


Note that the aTAM is a special case of the L-syncTAM where $l = 1$, and the syncTAM can be thought of as being a special case where $l$ is allowed to be $\infty$.

\subsection{Scaled shapes}

Given $S$, a connected set of points in $\mathbb{Z}^2$, we define a version of $S$ \emph{scaled by factor $c$}, and denote it by  $S^c$, as $S^c = \{(x,y) \in \mathbb{Z}^2 \mid (\lfloor \frac{x}{c} \rfloor, \lfloor \frac{y}{c} \rfloor) \in S\}$. Intuitively, $S^c$ is a version of $S$ expanded by replacing each point of $S$ by a $c \times c$ square of points. Likewise, for a set $A$ of shapes, $A^c$ is defined as $\{ S^c \mid S \in A \}$. 

Given two TASs $\calT$ and $\calT'$, if $\shapeasm{\calT} = \shapeasm{\calT'}$, then we say that $\calT$ \emph{shape-matches} $\calT'$. Similarly, if $\shapeasm{\calT}^c = \shapeasm{\calT'}^d$ for some $c, d \in \mathbb{Z}^+$, we say that $\calT$ \emph{scale shape-matches} $\calT'$.




\subsection{Synchronous intrinsic simulation}






    


Intrinsic universality via intrinsic simulation has been studied in a variety of tile-assembly models (e.g. \cite{IUSA,2HAMIU,DirectedNotIU,DDDIU,TempOneNotIU,DDDxModelSimICALP}). Essentially, intrinsic simulation deals with one system simulating the behavior of another modulo scale factor, with squares of tiles in the simulating system representing individual tiles of the simulated system. Intrinsic universality (IU) for a model (or class of systems within a model) occurs when there is a single tile set (along with associated functions, to be discussed) that can be used to create a seed structure from its tiles for any of the systems in the model (or class) and then intrinsically simulate them. Now we provide formal definitions which are largely taken from \cite{DDDxModelSimICALP} but are modified to incorporate the dynamic of synchronicity of the L-syncTAM.


From this point on, let $T$ be a tile set and let the scale factor be $m\in\Z^+$.
An \emph{$m$-block macrotile} over $T$ is a partial function $\alpha : \Z_m^2 \dashrightarrow T$, where $\Z_m = \{0,1,\ldots,m-1\}$.
Let $B^T_m$ be the set of all $m$-block macrotiles over $T$.
The $m$-block with no domain is said to be $\emph{empty}$.
For a general assembly $\alpha:\Z^2 \dashrightarrow T$ and $(x',y')\in\Z^2$, define $\alpha^m_{(x',y')}$ to be the $m$-block macrotile defined by $\alpha^m_{(x',y')}(i_x,i_y) = \alpha(mx'+i_x,my'+i_y)$ for $0 \leq i_x,i_y< m$.
For some tile set $S$, a partial function $R: B^{S}_m \dashrightarrow T$ is said to be a \emph{valid $m$-block macrotile representation} from $S$ to $T$ if for any $\alpha,\beta \in B^{S}_m$ such that $\alpha \sqsubseteq \beta$ and $\alpha \in \dom R$, then $R(\alpha) = R(\beta)$.

For a given valid $m$-block macrotile representation function $R$ from tile set~$S$ to tile set $T$, define the \emph{assembly representation function}\footnote{Note that $R^*$ is a total function since every assembly of $S$ represents \emph{some} assembly of~$T$; the functions $R$ and $\alpha$ are partial to allow undefined points to represent empty space.}  $R^*: \mathcal{A}^{S} \rightarrow \mathcal{A}^T$ such that $R^*(\alpha') = \alpha$ if and only if $\alpha(x,y) = R\left(\alpha'^m_{(x,y)}\right)$ for all $(x,y) \in \Z^2$.
For an assembly $\alpha' \in \mathcal{A}^{S}$ such that $R^*(\alpha') = \alpha$, $\alpha'$ is said to map \emph{cleanly} to $\alpha \in \mathcal{A}^T$ under $R^*$ if for all non empty blocks $\alpha'^m_{(x,y)}$, $(x,y)+(u_x,u_y) \in \dom(\alpha)$ for some $(u_x,u_y) \in \{(0,\pm 1), (\pm 1, 0)\}$, or if $\alpha'$ has at most one non-empty $m$-block $\alpha^m_{0,0}$.  In other words, $\alpha'$ may have tiles in macrotile blocks representing empty space in $\alpha$, but only if that position is adjacent to a macrotile that represents a tile in $\alpha$.  We call such growth ``around the edges'' of $\alpha'$ \emph{fuzz} and thus restrict it to be adjacent to only valid macrotiles, but not diagonally adjacent (i.e.\ we do not permit \emph{diagonal fuzz}).

In the following definitions, let $\mathcal{T} = \left(T, \sigma_T, \tau_T, l_T\right)$  and $\mathcal{S} = \left(S, \sigma_S, \tau_S, l_S\right)$ each be an L-syncTAS, and let $R$ be an $m$-block representation function $R:B^S_m \rightarrow T$. We will define how $\calS$ \emph{synchronously simulates} $\calT$ at \emph{spatial scale factor} $m$ and \emph{temporal scale factor} $s \in \mathbb{Z}^+$.

\begin{definition}\label{def:equiv-assemblies}
    Let $\calL = (T, \sigma, \tau, l)$ be an L-syncTAS, and let $\alpha$ and $\beta$ be two $\tau$-stable assemblies in $\calL$.  We say that $\beta$ is \emph{equivalent modulo $l$} to $\alpha$, and write $\beta \approx_l \alpha$, if and only if $\exists \gamma$ such that $\gamma$ is a $\tau$-stable assembly in $\calL$ and the following conditions hold:
    \begin{enumerate}
        \item $\alpha \rightarrow^\calL_1 \gamma$

        \item $\alpha \sqsubseteq \beta \sqsubseteq \gamma$

        \item $\exists i$, where $0 \le i < l$, such that $|\dom(\alpha)| + i = |\dom(\beta)| = |\dom(\gamma)| - l + i$

        \item $\dom(\beta) \setminus \dom(\alpha) \subseteq \partial^\calL(\alpha)$
    \end{enumerate}
\end{definition}

Essentially, if $\beta \approx_l \alpha$, $\beta$ can be thought of as an ``intermediate'' assembly between $\alpha$ and $\gamma$ (where $\alpha$ produces $\gamma$ in one step) where only $i$ (which is $< l$) tiles have been added to $\alpha$. While such an assembly $\beta$ is not producible in $\calL$ if $i> 0$, it will be useful to define such an assembly in the context of a simulator that may require many tile additions to simulate each individual tile addition of $\calL$. 

\begin{definition}
\label{def-equiv-prod} We say that $\mathcal{S}$ and $\mathcal{T}$ have \emph{equivalent productions} (under $R$), and we write $\mathcal{S} \Leftrightarrow \mathcal{T}$ if the following conditions hold:
\begin{enumerate}
        
        \item $\forall \alpha' \in \prodasm{S}, \exists \alpha \in \prodasm{T}$ such that $R^*(\alpha') \approx_{l_\mathcal{T}} \alpha$

        \item $\forall \alpha \in \prodasm{T}, \exists \alpha' \in \prodasm{S}$ such that $R^*(\alpha') = \alpha$
        
        \item $\left\{R^*(\alpha') | \alpha' \in \termasm{\mathcal{S}}\right\} = \termasm{\mathcal{T}}$.
        \item For all $\alpha'\in \prodasm{\mathcal{S}}$, $\alpha'$ maps cleanly to $R^*(\alpha')$.
\end{enumerate}
\end{definition}

\begin{definition}

\label{def:t-follows-s} We say that $\mathcal{T}$ \emph{follows} $\mathcal{S}$ (under $R$), and we write $\mathcal{T} \dashv_R \mathcal{S}$, if $\alpha' \rightarrow^\mathcal{S} \beta'$ for some $\alpha',\beta' \in \prodasm{\mathcal{S}}$ implies that $\exists \alpha'', \beta'' \in \prodasm{\calS}$ such that $R^*(\alpha') \approx_l R^*(\alpha'')$, $R^*(\beta') \approx_l R^*(\beta'')$ and $R^*(\alpha'') \to^\mathcal{T} R^*(\beta'')$.
\end{definition}

\begin{definition}\label{def:t-sync-follows-s}
We say that $\calT$ \emph{synchronously follows} $\calS$ (under $R$ with temporal scale factor $s$), and we write  $\mathcal{T} \dashv^s_R \mathcal{S}$, if $\mathcal{T} \dashv_R \mathcal{S}$ and, for every assembly sequence $\vec{\beta} = \beta_0, \beta_1, \beta_2, \ldots$ in $\calS$, there exists assembly sequence $\vec{\alpha} = \alpha_0, \alpha_1, \alpha_2, \ldots$ in $\calT$ such that: 
    \begin{enumerate}
        \item $\alpha_0 = \sigma$ and $R^*(\beta_0) = \sigma$
        
        \item For every $0 \le i < |\vec{\alpha}|$, there is $i' \in [i \cdot (s - 1), i \cdot (s+1)]$ such that $R^*(\beta_{i'}) = \alpha_i$

    \end{enumerate}
\end{definition}

The next definition essentially specifies that every time $\mathcal{S}$ simulates an assembly $\alpha \in \prodasm{\mathcal{T}}$, there must be at least one valid growth path in $\mathcal{S}$ for each of the possible next steps that $\mathcal{T}$ could make from $\alpha$ which results in an assembly in $\mathcal{S}$ that maps to that next step. While this definition is unfortunately dense, it accommodates subtle situations such as where $\calS$ must ``commit to'' a subset of possible representations in $\mathcal{T}$ before being explicitly mapped, under $R$, to any one in particular.

\begin{definition}\label{def:s-models-t}
We say that $\mathcal{S}$ \emph{models} $\mathcal{T}$ (under $R$), and we write $\mathcal{S} \models_R \mathcal{T}$, if for every $\alpha \in \prodasm{\mathcal{T}}$, there exists $\Pi \subset \prodasm{\mathcal{S}}$ where $\Pi \neq \emptyset$ and  $R^*(\alpha') = \alpha$ for all $\alpha' \in \Pi$, such that, for every $\beta \in \prodasm{\mathcal{T}}$ where $\alpha \rightarrow^\mathcal{T} \beta$, (1) for every $\alpha' \in \Pi$ there exists $\beta' \in \prodasm{\mathcal{S}}$ where $R^*(\beta') = \beta$ and $\alpha' \rightarrow^\mathcal{S} \beta'$, and (2) for every $\alpha'' \in \prodasm{\mathcal{S}}$ where $\alpha'' \rightarrow^\mathcal{S} \beta'$, $\beta' \in \prodasm{\mathcal{S}}$, $R^*(\alpha'') \approx_{l_\calT} \alpha$, and $R^*(\beta') = \beta$, there exists $\alpha' \in \Pi$ such that $\alpha' \rightarrow^\mathcal{S} \alpha''$.
\end{definition}

\begin{definition}\label{def:s-simulates-t}
We say that $\mathcal{S}$ \emph{synchronously simulates} $\mathcal{T}$ (under $R$ and temporal scale factor $s$) if $\mathcal{S} \Leftrightarrow_R \mathcal{T}$ (equivalent productions), $\mathcal{T} \dashv^s_R \mathcal{S}$ and $\mathcal{S} \models_R \mathcal{T}$ (equivalent dynamics).
\end{definition}

\section{Shapes}
\label{sec:shapes}

In this section we show how the syncTAM, even at temperature-1, is capable of forms of cooperative behavior and detection of empty locations. This allows the syncTAM to implement a flagpole, the shibboleth shape for cooperation, at temperature 1. This flagpole relies on \emph{implementing} cooperation thanks to synchronization by a \emph{cooperation gadget}. Synchronization also allows for parts of the syncTAM assembly to ``probe'' regions of space in order to assess whether or not some other part of the assembly has / has not arrived in that region within a fixed window of time. As such, syncTAM spaces may use empty space as input for computation.
This sidesteps the gap between weak and strict self-assembly of shapes \cite{ScaledPierFractals,FractalsSelfAssembleSODA,fractalgorithmica}, and enables the strict self-assembly of a (modified) discrete Sierpinski Triangle. Last, the cooperation gadget proves to be enough to build the same assemblies as any \emph{directed} aTAM system with a non-cooperative system running in the syncTAM.

First, we show how temperature-1 syncTAM systems can simulate cooperative attachments.

\subsection{Simulating a flagpole at temperature-1}
\label{sec:flagpole}

From the earliest investigations on temperature 1 and intrinsic universality\cite{temp1notIU}, it has been known that the aTAM temperature 1 system cannot simulate systems which exhibit cooperation, i.e. systems with a temperature at least $2$. In particular, let $\mathcal{F}$ be an aTAM system at temperature 2 as shown in \cref{fig:flagpole}, which is a slight variation of the one in \cite{temp1notIU}.

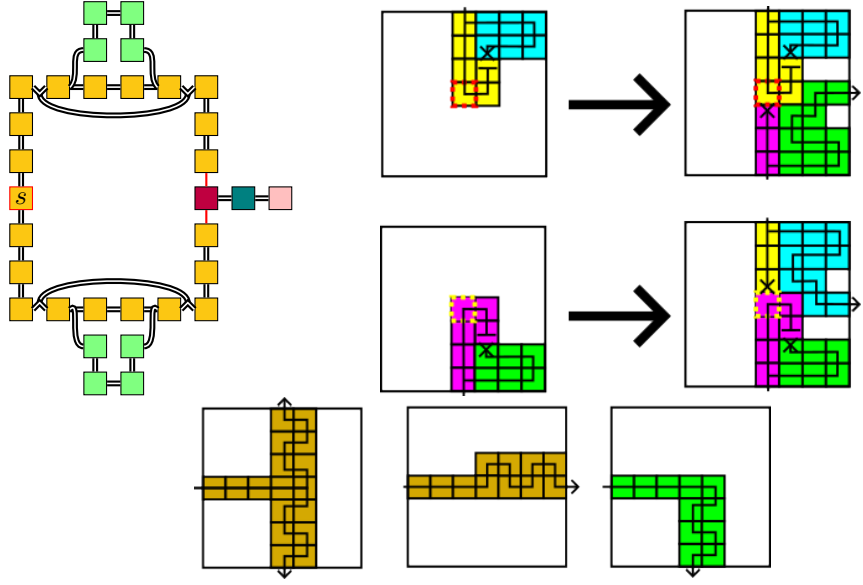
\begin{figure}
    \centering
  \begin{subfigure}[t]{.3\textwidth}

        \begin{tikzpicture}[tile/.style={draw, regular polygon sides=4, minimum size=.3cm, inner sep=0}]

  \node[red, tile, text=black, fill=lipicsYellow] (seed) {$s$};
\coordinate (u0) at (seed.north);
\coordinate (d0) at (seed.south);

\foreach \i [remember=\i as \previ (initially 0)] in {1, 2, 3} {
    \node[tile, fill=lipicsYellow, above=.5em of u\previ] (u\i) {};
    \node[tile, fill=lipicsYellow, below=.5em of d\previ] (d\i) {};
    \draw[thick, double] (u\previ) -- (u\i);
    \draw[thick, double] (d\previ) -- (d\i);
}

\coordinate (t0) at (u3.east);
\coordinate (b0) at (d3.east);

\foreach \i [remember=\i as \previ (initially 0)] in {1,...,5} {
    \node[tile, fill=lipicsYellow, right=.5em of t\previ] (t\i) {};
    \node[tile, fill=lipicsYellow, right=.5em of b\previ] (b\i) {};
}

\foreach \i [remember=\i as \previ (initially 5)] in {6,7} {
    \node[tile, fill=lipicsYellow, below=.5em of t\previ] (t\i) {};
    \node[tile, fill=lipicsYellow, above=.5em of b\previ] (b\i) {};
}

\node[tile, fill=green!50!white, above=.5em of t2] (t2a) {};
\node[tile, fill=green!50!white, above=.5em of t2a] (t2b) {};
\node[tile, fill=green!50!white, above=.5em of t3] (t3a) {};
\node[tile, fill=green!50!white, above=.5em of t3a] (t3b) {};

\node[tile, fill=green!50!white, below=.5em of b2] (b2a) {};
\node[tile, fill=green!50!white, below=.5em of b2a] (b2b) {};
\node[tile, fill=green!50!white, below=.5em of b3] (b3a) {};
\node[tile, fill=green!50!white, below=.5em of b3a] (b3b) {};

\node[tile, fill=purple, below=.5em of t7] (k) {};
\node[tile, fill=teal, right=.5em of k] (p) {};
\node[tile, fill=pink, right=.5em of p] (f) {};

\foreach \i [remember=\i as \previ (initially 1)] in {2,...,4} {
    \draw[thick, double] (b\previ) -- (b\i);
    \draw[thick, double] (t\previ) -- (t\i);
}
\foreach \i [remember=\i as \previ (initially 5)] in {6,7} {
    \draw[thick, double] (b\previ) -- (b\i);
    \draw[thick, double] (t\previ) -- (t\i);
}


\coordinate (t1west) at (t1.west);
\coordinate[shift={(0,-0.25em)}] (tleft) at (barycentric cs:t0=1,t1west=1);
\coordinate (t4east) at (t4.east);
\coordinate (t5west) at (t5.west);
\coordinate[shift={(0,-0.25em)}] (tright) at (barycentric cs:t5west=1,t4east=1);
\draw[thick, double] (tright) to[out=-90, in=-90, looseness=.5] (tleft);
\draw[thick, double] (t0) -- (tleft) -- (t1.west);
\draw[thick, double] (t4.east) -- (tright) -- (t5.west);
\draw[thick, double] (t1.east) to[in=180] (t2a.west);
\draw[thick, double] (t2a) -- (t2b) -- (t3b) -- (t3a); %
\draw[thick, double] (t3a.east) to[out=0] (t4.west);

\coordinate (b1west) at (b1.west);
\coordinate[shift={(0,0.25em)}] (bleft) at (barycentric cs:b0=1,b1west=1);
\coordinate (b4east) at (b4.east);
\coordinate (b5west) at (b5.west);
\coordinate[shift={(0,0.25em)}] (bright) at (barycentric cs:b5west=1,b4east=1);
\draw[thick, double] (bright) to[out=90, in=90, looseness=.5] (bleft);
\draw[thick, double] (b0) -- (bleft) -- (b1.west);
\draw[thick, double] (b4.east) -- (bright) -- (b5.west);
\draw[thick, double] (b1.east) to[in=180] (b2a.west);
\draw[thick, double] (b2a) -- (b2b) -- (b3b) -- (b3a); %
\draw[thick, double] (b3a.east) to[out=0] (b4.west);

\draw[thick, double] (k) -- (p) -- (f);
\draw[thick, red] (t7) -- (k) -- (b7);

\end{tikzpicture}
      \label{fig:flagpole-temp2}
    \end{subfigure}\quad\quad
    \begin{subfigure}[t]{.5\textwidth}
        \includegraphics[width=2.5in]{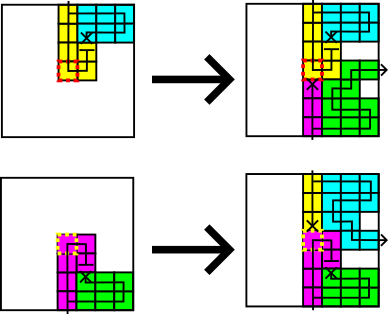}
    \end{subfigure}
    \begin{subfigure}[b]{.5\textwidth}
        \includegraphics[width=3.0in]{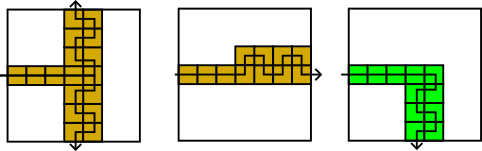}
    \end{subfigure}
    \caption{Our variation $\mathcal{F}$ on the ``Flagpole'' system from~\cite{temp1notIU}, with the TAS on the left. Then, going clockwise, some growth patterns. The first two are in the cooperative gadget (red in the figure to the left), with either the top or bottom neighboring tile arriving first, then the other one arriving and initiating the flagpole growth. Below, the macrotiles for the non-cooperative parts: all paths through the macrotiles are of length 11.}
    \label{fig:flagpole}
\end{figure}

As in~\cite{temp1notIU}, this system features two arms growing from left to right, one at the top on line $y = 3$, and one at the bottom on line $y = -3$. Each of these arms will turn back towards the $y = 0$ line at a non-deterministic position. If both arms do so in the same column, a three tile flag (in purple, teal and pink) will grow from their meeting point. Our variation adds green tiles, which allow the arms to dally while they move right: for every 4 column they advance, they can do so either straight with the yellow tiles, in 4 time steps, or with a detour through the green tiles, which takes them 6 time steps total. Because of this, even if both arms pick the same column $x = c$ to turn back to the center, they do not necessarily reach position $(c, \pm 1)$ at the same time.

Let $F$ be the set of shapes assembled by $\mathcal{F}$ in the aTAM at temperature 2. Note that in the syncTAM at temperature 2, $\mathcal{F}$ assembles the same set of shapes.

\begin{lemma}
    The TAS $\mathcal{F}$ cannot be scale shape-matched at temperature $1$ in the aTAM: there are no integer $c \geq 1$ and TAS $\mathcal{F}_1$ at temperature $1$ such that $\termasm{\mathcal{F}_1} = c \cdot F$.\label{lem:flagpole}
\end{lemma}

\begin{proof}
Since the system $\mathcal{F}$ is a trivial modification of the system in \cite{temp1notIU}, proving \cref{lem:flagpole} requires only a trivial application of the Window Movie Lemma that was used to show that the original system could not be simulated because any system attempting to simulate it would have to be able to grow assemblies of invalid shapes.
\end{proof}

On the other hand, there is a temperature 1 TAS $\mathcal{F}_\textrm{sync}$ which simulates $\mathcal{F}$ in the syncTAM, and thus has $\termasm{\mathcal{F}_\textrm{sync}} = c \cdot F$. A depiction of the system with $c = 7$ and the gadgets that simulate the cooperative behavior of the tile at the base of the flagpole can be seen in \cref{fig:flagpole}.



\subsection{A modified Sierpinski triangle at temperature-1}

In this section we show how a modified version of the Sierpinski triangle at scale factor $3$ can self-assemble in the syncTAM. This construction is used to demonstrate the ability of systems in the syncTAM to have gadgets that detect the lack of growth in an area, even without glue cooperativity (i.e. at temperature-1) and continue growth only in that absence of growth. 

Let $\Delta$ be the shape defined in \cref{fig:syncTAM-Sierpinski}.

\begin{theorem}\label{thm:sync-shape-temp1}
    The set $\Delta$ strictly self-assembles in the syncTAM at temperature 1 and scale 3 but does not strictly self-assemble in the aTAM at any temperature and scale.
\end{theorem}

\begin{proof}
    \Cref{lem:triangle-strict} shows that $\Delta$ does strictly self-assemble at temperature 1 in the syncTAM.

    To see that it does not strictly self-assemble in the aTAM, notice that like the actual Sierpinski Triangle, for each $k$ this shape features vertical lines where it has groups of tiles of constant size (here, 2) separated by gaps of size more than $k$, yet for no vector $\vec{p}$ does it have $\vec{p}$-periodic, square, non-empty subpatterns of arbitrary size. Then, by the Tree Pump Theorem~\cite{fractalgorithmica}, it cannot be strictly self-assembled in the aTAM. The proof is the same as that of Lemma 31 in that paper, as only these hypothesis on the shape is used.
\end{proof}

\begin{restatable}{lemma}{triangleStrict}\label{lem:triangle-strict}
    There is a syncTAM system which strictly self-assembles $\Delta$ at temperature 1 and scale-factor 3.
\end{restatable}

    We prove \cref{lem:triangle-strict} by providing a syncTAM system which strictly self-assembles the pattern.

    \begin{figure}
        \centering
        \includegraphics[width=0.8\textwidth]{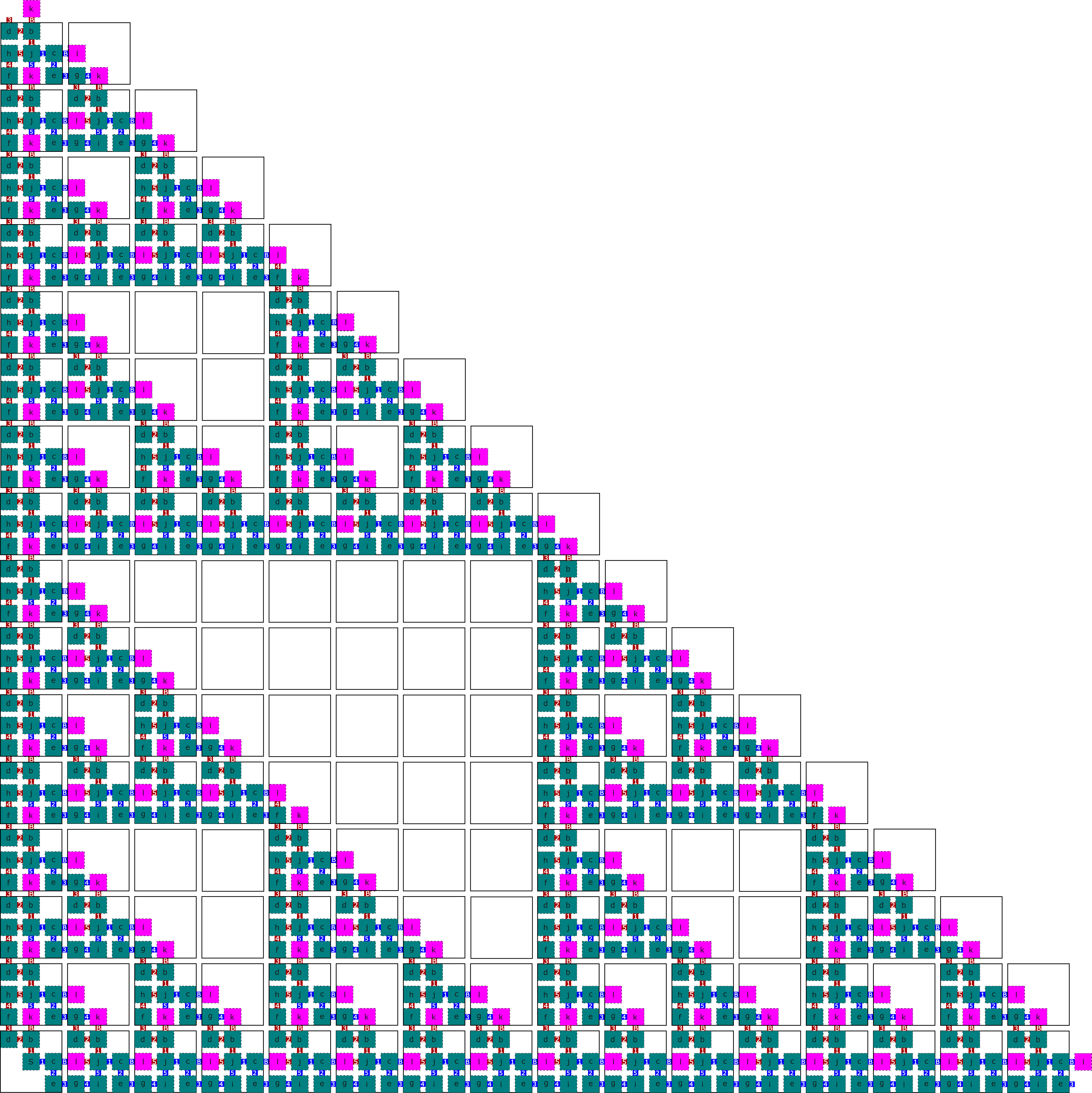}
        \caption{A portion of the infinite modified Sierpinski triangle that strictly self-assembles in the syncTAM at temperature 1.\label{fig:syncTAM-Sierpinski}}
    \end{figure}

    A portion of the infinite shape can be seen in \cref{fig:syncTAM-Sierpinski}, and the syncTAM system can be downloaded at \url{http://self-assembly.net/wiki/index.php/Synchronous_Tile_Assembly_Model_(syncTAM)}, which also provides a link to the simulator for simulating the system. Our syncTAM system uses $3 \times 3$ \emph{macrotiles} to logically ``simulate'' the behaviors of the aTAM tiles.
    The synchronous nature of the syncTAM allows us to use the timing of different portions of the assembly's growth to guarantee that growth will be blocked from proceeding whenever growth of a stage in one direction reaches the correct length because it crashes into growth from the other direction that has also reached the correct length. In the aTAM, it cannot be guaranteed that the growth of both portions will proceed at identical rates, and thus the necessary blocking may not occur. However, by adding tiles in every frontier location at every step the syncTAM provides such a guarantee and thus provides a powerful tool for designing systems. A detailed description of the system can be found in the appendix, in \cref{apdx:Sierpinski}.

\section{Matching Assemblies of Directed aTAM Systems With Non-cooperative syncTAM Systems}\label{sec:directed}

\begin{figure}[t]
  \centering
    \begin{tikzpicture}[xshift=-3cm]

\clip (-4.2cm, -.6cm) rectangle (9.5cm, 4.5cm);

\begin{scope}[xshift=5cm, scale=.2]
  \fill[lipicsYellow] (0,0) rectangle +(1,1);
  
  \foreach \i in {0,...,20} {
    \fill[orange!50!black] (\i+1,\i) rectangle +(1,1);
    \fill[orange] (\i+1,\i+1) rectangle +(1,1);
    \fill[green!50!black] (-\i-1,\i) rectangle +(1,1);
    \fill[green] (-\i-1,\i+1) rectangle +(1,1);
  }

  \tikzmath{
    {
      \fill[blue] (4, 5) rectangle +(1,1);
      \fill[red] (-3, 4) rectangle +(1,1);
      \fill[cyan] (0, 5) rectangle +(1,1);
      \fill[cyan] (1, 5) rectangle +(1,1);
      \fill[cyan] (2, 5) rectangle +(1,1);
      \fill[cyan] (3, 5) rectangle +(1,1);
      \fill[cyan] (2, 4) rectangle +(1,1);
      \fill[red] (1, 4) rectangle +(1,1);
      \fill[purple] (-1, 5) rectangle +(1,1);
      \fill[purple] (-2, 4) rectangle +(1,1);
      \fill[purple] (-1, 4) rectangle +(1,1);
      \fill[purple] (0, 4) rectangle +(1,1);
    };
    for \k in {8,12,...,20} {
      {
        \fill[blue] (\k, \k+1) rectangle +(1,1);
        \fill[red] (-\k+1, \k) rectangle +(1,1);
      };
      for \x in {0,4,...,\k-4} {
        {
          \fill[cyan] (\x, \k+1) rectangle +(1,1);
          \fill[cyan] (\x+1, \k+1) rectangle +(1,1);
          \fill[cyan] (\x+2, \k+1) rectangle +(1,1);
          \fill[cyan] (\x+3, \k+1) rectangle +(1,1);
          \fill[cyan] (\x+2, \k) rectangle +(1,1);
          \fill[red] (-\x+1, \k) rectangle +(1,1);
          \fill[purple] (-\x-1, \k+1) rectangle +(1,1);
          \fill[purple] (-\x-2, \k) rectangle +(1,1);
          \fill[purple] (-\x-1, \k) rectangle +(1,1);
          \fill[purple] (-\x, \k) rectangle +(1,1);
        };
      };
    };
  }

\end{scope}

\begin{scope}[xshift=-1.5cm,tile/.style={fill=black, draw, regular polygon sides=4, minimum size=.3cm, inner sep=0, font=\footnotesize}, interrupt/.style={fill=white,inner sep=1pt}]

  
  \node[tile, fill=lipicsYellow] (seed) {$s$};

  \coordinate (tl0) at (seed.west);
  \coordinate (tr0) at (seed.east);

  \foreach \i [remember=\i as \previ (initially 0)]in {1, ..., 4} {
    \node[tile, fill=green!50!black, left=.6em of tl\previ] (bl\i) {};
    \node[tile, fill=green, above=.5em of bl\i] (tl\i) {};

    \draw[thick] (tl\previ.west) -- (bl\i.east);
    \draw[thick] (bl\i.north) -- (tl\i.south);

    \node[tile, fill=orange!50!black, right=.6em of tr\previ] (br\i) {};
    \node[tile, fill=orange, above=.5em of br\i] (tr\i) {};
    \draw[thick] (tr\previ.east) -- (br\i.west);
    \draw[thick] (br\i.north) -- (tr\i.south);
  }

  \node[tile, right=1em of tl4, fill=red] (e0) {};
  \node[tile, above=1em of tr4, fill=blue] (w0) {};

  \foreach \i [remember=\i as \previ (initially 0)]in {1, ..., 4} {
    \node[tile, left=.5em of w\previ, fill=cyan] (w\i) {};
    \draw[thick] (w\previ) -- (w\i);
  }
  \node[tile, below=.5em of w2, fill=cyan] (ws) {};

  \foreach \i [remember=\i as \previ (initially 0)]in {1, ..., 3} {
    \node[tile, right=.5em of e\previ, fill=purple] (e\i) {};
    \draw[thick] (e\previ) -- (e\i);
  }

  \node[tile, above=.5em of e2, fill=purple] (en) {};

  \begin{scope}[thick]
    \draw[thick] (e2) -- (en);
    \draw[thick] (w2) -- (ws);
    \draw (tl4.west) to[out=-150, in=-45, looseness=2] (bl1.east);
    \draw (tr4.east) to[out=-30, in=-135, looseness=2] (br1.west);
    \draw (tl4.east) to[out=0, in=180] (e0.west);
    \draw (tr4.north) to (w0.south);
    \draw (e3.east) to[out=-30, in=-150, looseness=1.25] (e0.west);
    \draw (w4.west) to[out=150, in=30, looseness=1.25] (w1.east);
  \end{scope}
\end{scope}
\end{tikzpicture}
    \caption{A system which is directed in the syncTAM but not in the aTAM. The tiles of the temperature $1$ syncTAM system $\calS$ (the seed is the yellow $s$ tile), and a portion of the infinite shape $V$ assembled by $\calS$ in the syncTAM. Growth begins from the seed tile at the bottom and proceeds with the diagonal ``arms'' growing upward to the left and right, and with red rows growing right from the left arm, and blue rows growing left from the right arm, at every 4th step. The green part is $L$, the orange and brown part is $R$, the blue parts are the $B^+_k \cup T^+_k$, while the red parts are the $B^-_k \cup T^-_k$.}
    \label{fig:impossible-V}
\end{figure}
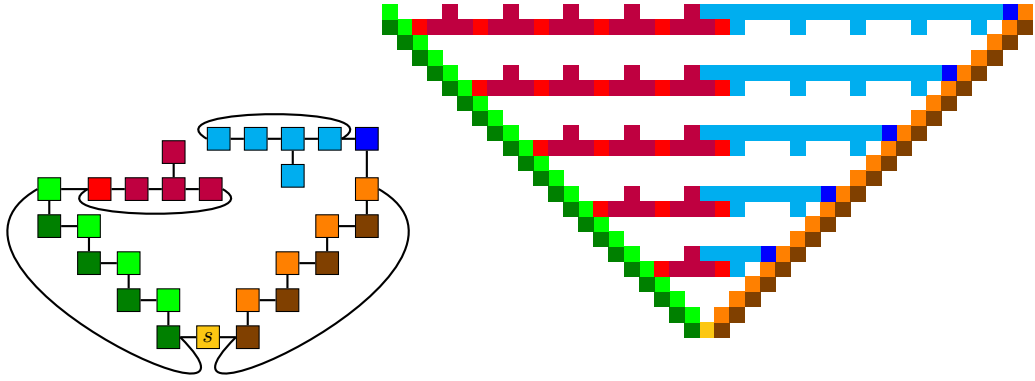

Any system which is directed in the aTAM is also directed in the syncTAM, since any syncTAM production must be an aTAM production~\cite{syncTAM}. In this directed case, it is actually possible to assemble a scaled version of the aTAM system's terminal assembly. 
Note that being directed in the aTAM is a stronger condition than being directed in the syncTAM: in order to be directed in the aTAM, a TAS needs to be deterministic in the face of asynchronism of tile additions. For example, the system in \cref{fig:impossible-V} is directed in the syncTAM, assembling the production depicted on the figure, while in the aTAM, the meeting point of the two parts of each rung varies based on the particular assembly sequence chosen.

\begin{restatable}{theorem}{directedSim}
\label{thm:directedSim}%
    If $\calS$ is directed in the aTAM, there is a temperature $1$ syncTAM system $\calS'$ which assembles a scaled version of $\alpha \in \termasm{S}$, for some scale factor $s$ and macrotile representation function $R$. 
\end{restatable}

The proof of this theorem can be found  in \cref{apdx:directedSimProof} of the appendix. The construction leverages the cooperation gadget of \cref{sec:flagpole}, using the directedness of $\calS$ to argue away timing problems due to the need to simulate the asynchronicity of $\calS$.

\section{L-SyncTAM Results}
\label{sec:LsyncTAM}

First we present \cref{lem:indep}, which will be useful for several of the proofs in this section. \cref{lem:indep} shows that when a certain set of requirements are met for subassemblies produced by systems in the aTAM or L-syncTAM, all such subassemblies grow independently of each other.

Given a shape $S$ let $\texttt{bound}(S)$ refer to the set of unit vectors in $\mathbb{Z}^2$ that form its perimeter, or \emph{boundary}. Given two points $l_1, l_2 \in \mathbb{Z}^2$, let $\texttt{dist}(l_1,l_2)$ be a function that returns the minimal Manhattan distance between points $l_1$ and $l_2$.

\begin{lemma}[Eventually Independent Subassemblies]\label{lem:indep}
    Let $\calT = (T, \sigma, \tau, l)$ be a shape-directed L-syncTAS in the L-syncTAM for any $l \ge 1$ and let $S$ be the shape that $\calT$ strictly self-assembles. If there exists a cut $C$ of $\mathbb{Z}^2$ that separates $S$ into shapes $s_0, s_1, \ldots, s_n$ such that:
    \begin{enumerate}
        \item $\forall 0 \le i,j \le n$ where $i \neq j$, $s_i \cap s_j = \emptyset$, [They are all disjoint]
        
        \item $\bigcup_{i=0}^n s_i = S$, [Their union is $S$]
        
        \item $\dom \sigma \subseteq s_0$, [The seed is completely contained in $s_0$]
        
        \item $\forall 0 < i \le n$, $|s_i| = \infty$, [Other than perhaps $s_0$, all must be infinite]
        
        \item For each $0 < i \le n$, $|C \cap \texttt{bound}(s_i)| < \infty$, [$C$ has a finite intersection with each except maybe $s_0$]

        \item For all $0 < i \neq j \le n$, $\forall l_1 \in s_i$, $\forall l_2 \in s_j$, $\texttt{dist}(l_1,l_2) \ge 3$. [The minimum distance between points in any pair of shapes (not including $s_0$) is at least 3.]\label{item:fuzz}
        

        \item For each $0 < i \le n, C \cap \texttt{bound}(s_i) =\texttt{bound}(s_0) \cap\texttt{bound}(s_i)$, [$C$ encapsulates the entire shared boundary of each with $s_0$]
    \end{enumerate}
    then, the following holds: Let $\alpha \in \termasm{T}$ be any terminal assembly of $\calT$ and let $\vec{\alpha}$ be any arbitrary assembly sequence that produces $\alpha$. Recall that $\dom \alpha = S$. There exists some $f \in \mathbb{N}$ such that for every assembly step of $\vec{\alpha}$ whose index (i.e. location in the sequence) is greater than $f$, $\forall 0 < i \le n$ there is a frontier location inside of $s_i$.
\end{lemma}

\begin{proof}
    To prove \cref{lem:indep}, we first note that since the seed $\sigma$ is completely contained in $s_0$, no location of $s_i$ initially has a tile, and for all $0 < i \le n$, the first location in $s_i$ that receives a tile must be adjacent to cut $C$ since that cut separates the subassembly in $s_i$ (which we will call $\alpha_i$) from the rest of the assembly. For any (necessarily empty) location $\vec{l} \in \mathbb{Z}^2$ to become a frontier location, at least one neighboring location $\vec{l}'$ must have been a frontier location that became filled with a tile.
    If $\vec{l}$ first becomes a frontier location during some assembly step $s$, then some neighboring location must have received a tile during step $s-1$.
    This means that any location in $s_i$ that receives a tile must have been included in a path of successive, adjacent frontier locations that originated adjacent to the seed and at some point crossed $C$.

    By the conditions of the lemma, the cut $C$ only overlaps with a finite portion of the perimeter of $s_i$.
    Also, since each such $s_i$ has infinite size that is filled (in the limit) by tiles, there must be an infinite number of assembly steps during which there exists a frontier location within $s_i$ during $\vec{\alpha}$.
    For any assembly step $s$ during which there is not a frontier location in $s_i$, since there must exist a later step where a location in $s_i$ receives a tile, there must exist a later step $s' > s$ where there exists a frontier location adjacent to $C$, as this is the only boundary into which a path of frontier locations can enter $s_i$. Since the locations adjacent to $C$ are finite, this can only happen a finite number of times. Therefore, there exists some constant $f_i$ such that after the $f_i$th assembly step no more tiles are placed adjacent to $C$, and since the number of tiles placed within $s_i$ is infinite, for the infinite number of steps greater than $f_i$, there must be frontier locations within $s_i$.

    Finally, since the argument holds for each $s_i$, each has its own such constant $f_i$ and the $f$ of the lemma is simply the maximum of them. \cref{corr:indep_sim} directly follows from the fact that each shape $s_i$ for $0<i \le n$ is separated by a distance greater than that of fuzz.

\end{proof}

\begin{corollary}\label{corr:indep_sim}
Let $\calS = (T, \sigma, \tau, l)$ be an $L$-syncTAS in the $L$-syncTAM for any $l \ge 1$. If $\calS$ synchronously simulates some $\calT$ such that \cref{lem:indep} holds for $\calT$, then  \cref{lem:indep} holds for $\calS$.
\end{corollary}

\cref{corr:indep_sim} follows immediately from \cref{lem:indep} since the scaled shape assembled by $\calS$ must have a similarly scaled and translated cut $C'$ for which all of the same properties hold (noting that the condition that the distance required between shapes by \cref{item:fuzz} ensures sub-shapes are still disjoint even accounting for fuzz in the simulation), and therefore so does the proof.

\subsection{The class of L-syncTAM systems is not IU for any $l > 1$}

For any fixed $l \ge 2$, the class of L-syncTAM systems with synchronicity $l$ is not intrinsically universal (for synchronous simulation). For any fixed candidate simulator, there are systems that cannot be simulated by that simulator.

\begin{lemma}\label{lem:many-unscalables}
    For all $l \ge 2$, there is an infinite family $(\calS^l_n)_{n \in \mathbb{N}}$ of L-syncTAM systems with synchronicity $l$ such that for each $n$, $\calS^l_n$ has $n + 5$ tiles and cannot be synchronously simulated in the L-syncTAM at synchronicity $l$ with any temporal scale factor $s$ and any spatial scale factor greater than $1$.
\end{lemma}

\begin{proof}
    Fix $n$. Consider the system of \cref{fig:no-IU-sim}, with the green seed tile fixed at $(0,0)$, and $4$ tile types (yellow, gold, aqua, and blue) that grow a row with a periodic sequence of those $4$ tile types infinitely far to the right. From each blue tile, a periodic sequence of $n$ unique tile types attach (all shown as pink) infinitely often, to create infinite vertical columns. This system is directed, and strictly-self assembles the shape shown in \cref{fig:no-IU-sim}. 
    Call this system $\mathcal{S}^l_n = (T,\sigma,1,l)$ and note that $|T| = n + 5$.

    We will prove \cref{lem:many-unscalables} by contradiction, so assume that L-syncTAM system $\mathcal{S}'$ synchronously simulates $\mathcal{S}^l_n$ at scale factor $c > 1$ and some temporal scale factor $s$. We now inspect a particular assembly sequence in $\mathcal{S}^l_n$ and show that $\mathcal{S}'$ does not synchronously follow $\mathcal{S}^l_n$.

    Let $\vec{\alpha}$ be an assembly sequence in $\mathcal{S}^l_n$ that proceeds as follows. The horizontal row grows to a length $> 4((3s + 1)l + (l - 1)) + 2$ by starting from the seed and attaching $(3s+1)l + (l - 1)$ copies of the 4-tile repeating segment. From each blue tile, a vertical column begins to grow. The columns can grow in any order, resulting in $(3s+1)l + (l - 1)$ vertical columns which will each ultimately be infinite in height. Note that the horizontal row will continue to grow and initiate further upward columns, but we will not need to include them in our analysis.

    There must be an assembly sequence $\vec{\alpha}'$ in $\mathcal{S}'$ that models $\vec{\alpha}$, and let $\vec{\alpha}'$ proceed until it has not only grown sufficiently far horizontally (i.e. $> (c \cdot 4((3s + 1)l + (l - 1)) + 2)$), but also continued growth so that all subassemblies mapping to the $(3s+1)l + (l - 1)$ vertical arms are of sufficient height that \cref{corr:indep_sim} of \cref{lem:indep} can be applied, with cut $C$ separating the $(3s+1)l + (l - 1)$ vertical columns into shapes $s_1, s_2, \cdots, s_{(3s+1)l + (l - 1)}$ of the lemma, so that each vertical column is guaranteed to have a frontier location in all further assembly steps.
    The portion of the horizontal bar containing the seed will be $s_0$, and the rest of the assembly to the east of the cut will be $s_{(3s+1)l+l}$. We extend assembly sequence $\vec{\alpha}'$ as follows.
    
    We will refer to each vertical column as $V_i$ for $0 \le i < (3s+1)l + (l-1)$, and to the rightmost $l-1$ vertical columns of that group as \emph{buffer} columns. Since $\mathcal{S}'$ has synchronicity $l$, there must be $l$ tiles placed during every assembly step. Therefore, starting with $i=0$, we perform assembly steps such that each places one tile in each of the buffer columns (using $l-1$ of the tile placements), and one tile in $V_i$. We repeat this until there is a frontier location, which we will refer to as a \emph{nonresolving frontier location}, in $V_i$ so that when a tile is placed there, no new macrotile resolves to a tile of $\mathcal{S}^l_n$. That is, assuming the assembly $\alpha'$ in $\mathcal{S}'$ maps under $R$ to some assembly $\alpha$ in $\mathcal{S}^l_n$ (i.e. $R(\alpha') = \alpha$), when a tile is placed in that nonresolving frontier location to create assembly $\alpha''$ in $\mathcal{S}'$, it is still the case that $R(\alpha'') = \alpha$. This must be possible because the scale factor of the simulation is $c > 1$ and therefore macrotiles must have dimensions $> 1$, meaning that any path from the south to the north of a macrotile must be of length $> 1$ and a macrotile only resolves during a single assembly step. Thus, there must be at least one step during which there is a tile placed that does not cause the macrotile to resolve, and that tile's location is a nonresolving frontier location immediately prior to that. 

    Once it occurs that a nonresolving frontier location exists in $V_i$, we increment to $i+1$ and repeat until nonresolving frontier locations exist in each of the $(3s+1)l$ leftmost vertical columns. At this point, let the current assembly be $\beta'$, and let $\beta = R(\beta')$ be the assembly in $\mathcal{S}^l_n$ to which it maps. We now perform $3s+1$ assembly steps as follows. For $0 \le p < 3s+1$, each step consists of placing a single tile in each $V_q$, for $lp \le q < (l+1)p$, into the nonresolving frontier locations. In this way, each column in each set of $l$ columns starting from the left receives a single tile. However, since all locations were nonresolving frontier locations, that means that the assembly sequence proceeded $3s+1$ steps without any change in the assembly to which it maps in $\mathcal{S}^l_n$. This breaks \cref{def:t-sync-follows-s}, i.e. the simulated system does not synchronously follow the simulator because it must continue to add $l$ tiles per assembly step, but even allowing for the maximum amount of time within the temporal scaling of $s$, which is a window of size $3s$, $\mathcal{S}'$ resolves to no additional tiles. This is a contradiction that $\mathcal{S}'$ synchronously simulates $\mathcal{S}^l_n$ and thus \cref{lem:many-unscalables} is proven.
\end{proof}

\begin{theorem}\label{thm:k-sync-not-IU}
    For all $l \ge 2$, the class of L-syncTAM systems with synchronicity $l$ is not intrinsically universal.
\end{theorem}

\begin{figure}
    \centering
    \includegraphics[width=0.3\linewidth]{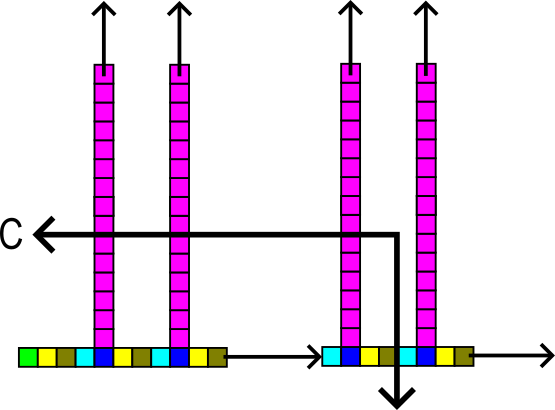}
    \caption{A schematic depiction of a TAS in the L-syncTAM with synchronicity $l$ that cannot be simulated by a system that is IU for the L-syncTAM at synchronicity $l$. The seed is the green tile on the bottom left. An infinite number of repetitions of $4$ tiles attach to its right. From each blue tile, a column grows upward infinitely, with repetitions of a series of $n$ unique tile types attaching. Cut $C$ of \cref{lem:indep} is also shown.}
    \label{fig:no-IU-sim}
\end{figure}

\begin{proof}
    Let $l \in \mathbb{N}$. Any fixed L-syncTAM system at synchronicity $l$ with $t$ tile types can only simulate systems with at most $(t+1)^s$ tiles at scale $s$, as that is the cardinality of the domain of the representation function. But by \cref{lem:many-unscalables}, there are infinitely many L-syncTAM systems at synchronicity $l$ which can only be simulated at scale 1. Hence, regardless of the size of the tile set of a fixed system simulator, since it must be finite, no fixed system can simulate all of them and the class of L-syncTAM systems with synchronicity $l$ is not intrinsically universal.
\end{proof}

\subsection{Shapes that cannot self-assemble with less synchronization}

We now present a result that shows that for any synchronicity value $l > 1$, there exists a shape which cannot be made by any system with lower synchronicity. This is a much stronger statement than simply saying there is a system with synchronicity $l$ that cannot be simulated by any with lower synchronicity. This instead shows that, even without the constraints of following a system's dynamics, the shape simply cannot be strictly produced.

For any $l > 1$, the shape pictured on \cref{fig:rays-and-cone} strictly self-assembles in an L-syncTAM system with synchronicity $l$, but not in any system with lower synchronicity.

\begin{restatable}{theorem}{noLessSync}

\label{thm:imposs-k-1-sync-shape}
    For each $l > 1$, there exists a shape $S_l$ such that:
    \begin{enumerate}
        \item There exists L-syncTAM system $\calL_l= (T, \sigma, 2, l)$ such that $\calL_l$ strictly self-assembles $S_l$.

        \item For all $l' < l$, there exists no L-syncTAM system $\calL' = (T', \sigma', \tau, l')$ and scale factor $c \in \mathbb{Z}^+$ such that $\calL'$ strictly self-assembles $c \cdot S_l$.
    \end{enumerate}

\end{restatable}

The proof of \cref{thm:imposs-k-1-sync-shape} relies on observing that any simulator for this system must have $l - 2$ frontier positions for the infinite arms growing upwards. On the other hand, by a pumping argument, it cannot encode the positions within the cone in order to be able to find the middle without keeping pairs of tile additions in sync. Hence synchronicity $\ge l$ is needed in order to assemble that shape. The full proof is in the appendix, in \cref{apdx:imposs-k-1-sync-shape}.

\begin{figure}
    \centering
    \includegraphics[width=0.8\linewidth]{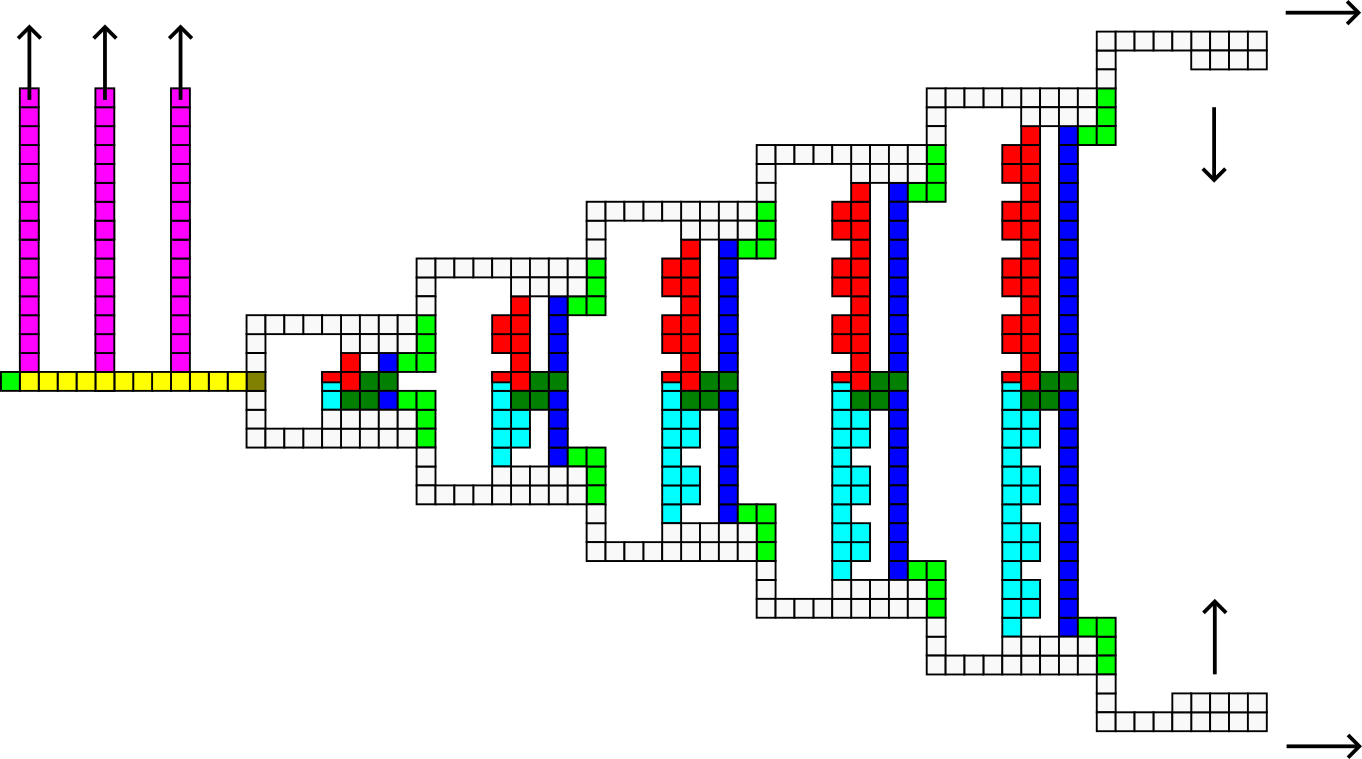}
    \caption{Initial portion of shape $S_5$ (i.e. $S_l$ for $l = 5$) from the proof of \cref{thm:imposs-k-1-sync-shape}. The seed tile (green) initiates growth of the yellow line which initiates the growth of $l - 2$ \texttt{arms} (pink) that each grow infinitely upward. The gold tile at the end of the yellow row initiates growth of the \texttt{cone} (white) that grows infinitely to the right, with \texttt{rungs} (red, aqua, green, blue, and lime green) that periodically grow between the sides of the \texttt{cone}. A depiction of the growth sequence of the \texttt{cone} can be seen in \cref{fig:cone-growth}. }
    \label{fig:rays-and-cone}
\end{figure}

\subsection{Systems that cannot be simulated by systems with greater synchronization}

The following theorem states that for any $l \ge 1$, there is a system in the L-syncTAM with synchronicity $l$ that cannot be simulated by any L-syncTAM system whose synchronicity $l'$ is greater than $l$, i.e. the additional synchronization causes any system to fail to correctly simulate.

\begin{theorem}\label{thm:no-up-sim}
    For each $l \ge 1$, there exists an L-syncTAM system $\calS_l$ with synchronicity $l$ such that, for any $l' > l$, there exists no $L$-syncTAM system $\calS'$ with synchronicity $l'$ which simulates $\calS_l$.
\end{theorem}

\begin{proof}
    We prove \cref{thm:no-up-sim} by first describing a directed $L$-syncTAS $\calS_l$, for arbitrary $l \ge 1$. A schematic depiction can be seen in \cref{fig:no-up-sim}. It consists of a green seed tile fixed at $(0,0)$, and 4 tile types (yellow, gold aqua, and blue) that grow a row with a periodic sequence of those 4 tile types infinitely far to the right. From each blue tile, a single tile (pink) attaches to itself infinitely often, to create infinite vertical columns. This system is similar to that which was shown in the proof of \cref{lem:many-unscalables}, with the only modification being that each infinite vertical column only contains a single repeating tile type. Call this system $\calS_l = (T, \sigma, 1, l)$ and note that $|T| = 6$.

    Now, let $l'$ be an integer $>l$. We will prove \cref{thm:no-up-sim} by contradiction, so assume that $\calS'$ is an $L$-syncTAS which synchronously simulates $\calS_l$ at some temporal scale factor $s$ and spatial scale factor $c$. Let us now inspect a particular assembly sequence of $\calS'$.

    Define $k = (3s+1)l' + (l' - 1)$, a constant whose value will become clear later.
    Let $\vec{\alpha}$ be an assembly sequence in $\calS_l$ that proceeds as follows. The horizontal row grows to a length $>4k+2$ by starting from the seed and attaching $k$ copies of the 4-tile repeating segment. From each blue tile, a vertical column begins to grow, resulting in $k$ vertical columns which will ultimately be infinite in height. Note that the horizontal row will continue to grow and initiate further upward columns, but we do not need to include them in our analysis. 

    There must be an assembly sequence $\vec{\alpha}'$ in $\calS'$ that models $\vec{\alpha}$. Let $\vec{\alpha}'$ proceed until it has not only grown sufficiently far horizontally (i.e. $> c \cdot (4k + 2)$), but also continues growth so that all subassemblies mapping to the $k$ vertical columns are of sufficient height that \cref{corr:indep_sim} of \cref{lem:indep} may be applied, with cut $C$ separating $k$ vertical columns into shapes $s_1, s_2, \ldots, s_k$, such that each vertical column is guaranteed to have a frontier location in all future assembly steps. The portion of the horizontal bar containing the seed will be $s_0$, and the rest of the assembly to the east of the cut will be $s_{k+1}$. We extend assembly sequence $\vec{\alpha}'$ as follows.

    We'll refer to each vertical column as $V_i$ for $0 \le i < k$, and to the rightmost $l'-1$ of those as \textit{buffer} columns. Since $\calS'$ has synchronicity $l'$, there must be $l'$ tiles placed during each assembly step. Therefore, starting with $i=0$, we perform  assembly steps such that each places one tile in each of the buffer columns (using $l'-1$ of the tile placements in each assembly step), and one in $V_i$. We repeat this until there is a frontier location, which we will refer to as a \textit{resolving frontier location}, in $V_i$ so that there exists a tile type in $\calS'$ that may be placed there whose attachment causes a new macrotile to resolve to a tile in $\calS_l$. We will refer to this tile as a \textit{resolving tile}. That is, assuming the assembly $\alpha'$ in $\calS'$ maps under $R$ to some assembly $\alpha$ in $\calS$ (i.e. $R(\alpha') = \alpha$), when a resolving tile is placed in that resolving frontier location to create an assembly $\alpha''$ in $\calS'$, it is the case that $R(\alpha'') \neq \alpha$.

    Once it occurs that a resolving frontier location exists in $V_i$, we increment to $i+1$ and repeat until resolving frontier locations exist in each of the $(3s+1)l'$ leftmost vertical columns. At this point, let the current assembly be $\beta'$, and let $\beta = R(\beta')$ be the assembly in $\calS$ to which it maps. We now perform $3s+1$ assembly steps as follows. For $0 \le p < 3s+1$, each step consists of placing a single resolving tile into the resolving frontier locations in each $V_q$ for $l'p \le q < (l'+1)p$. In this way, each column in each set of $l'$ columns starting from the left receives a single tile, causing $l'$ macrotiles to resolve.
    This breaks \cref{def:t-sync-follows-s}, i.e. the simulated system does not synchronously follow the simulator because the simulator adds $>l$ tiles in every assembly step for each of $3s+1$ consecutive steps, causing the dynamics to differ even beyond compensation by the temporal scaling factor $s$. This is a contradiction of the claim that $\calS'$ simulates $\calS_l$, and thus \cref{thm:no-up-sim} is proven.
    


\begin{figure}
    \centering
    \includegraphics[width=0.3\linewidth]{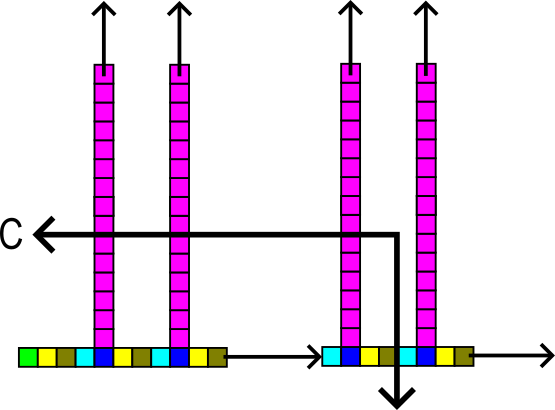}
    \caption{A depiction of a TAS in the L-syncTAM with synchronicity $l$ that cannot be simulated by any TAS in the L-syncTAM with synchronicity $l' > l$. The seed is the green tile on the bottom left, and an infinite number of repetitions of 4 tiles can attach to its right. From each blue tile, a column grows upward infinitely, with copies of the same pink tile type attaching. Cut $C$ of \cref{lem:indep} is also shown.}
    \label{fig:no-up-sim}
\end{figure}

\end{proof}

\bibliography{tam,experimental_refs,slats,synctam}

@article{fractalgorithmica,
  author       = {Florent Becker and
                  Daniel Hader and
                  Matthew J. Patitz},
  title        = {Strict Self-Assembly of Discrete Self-Similar Fractals in the Abstract
                  Tile Assembly Model},
  journal      = {Algorithmica},
  volume       = {88},
  number       = {1},
  pages        = {18},
  year         = {2026},
  url          = {https://doi.org/10.1007/s00453-025-01340-w},
  doi          = {10.1007/S00453-025-01340-W},
  bibsource    = {dblp computer science bibliography, https://dblp.org}
}

@inproceedings{syncTAM,
  TITLE = {{Synchronous Versus Asynchronous Tile-Based Self-Assembly}},
  AUTHOR = {Becker, Florent and Drake, Phillip and Patitz, Matthew and Rogers, Trent},
  URL = {https://hal.science/hal-05299696},
  BOOKTITLE = {{DNA 2025}},
  ADDRESS = {Lyon, France},
  PUBLISHER = {{Schloss Dagstuhl -- Leibniz-Zentrum f{\"u}r Informatik}},
  YEAR = {2025},
  MONTH = Aug,
  DOI = {10.4230/LIPIcs.DNA.31.9},
  HAL_ID = {hal-05299696},
  HAL_VERSION = {v1},
}

@inproceedings{FractalsSelfAssembleSODA,
  title={Strict Self-Assembly of Discrete Self-Similar Fractals in the abstract Tile Assembly Model},
  author={Becker, Florent and Hader, Daniel and Patitz, Matthew J},
  booktitle={Proceedings of the 2025 Annual ACM-SIAM Symposium on Discrete Algorithms (SODA'25), New Orleans, USA},
  pages={2387--2466},
  year={2025},
  organization={SIAM}
}

@inproceedings{DDDxModelSimICALP,
  author       = {Daniel Hader and
                  Matthew J. Patitz},
  editor       = {Kousha Etessami and
                  Uriel Feige and
                  Gabriele Puppis},
  title        = {The Impacts of Dimensionality, Diffusion, and Directedness on Intrinsic Cross-Model Simulation in Tile-Based Self-Assembly},
  booktitle    = {50th International Colloquium on Automata, Languages, and Programming,
                  {ICALP} 2023, July 10-14, 2023},
  address      = {Paderborn, Germany},
  series       = {LIPIcs},
  volume       = {261},
  pages        = {71:1--71:19},
  publisher    = {Schloss Dagstuhl - Leibniz-Zentrum f{\"{u}}r Informatik},
  year         = {2023},
  url          = {https://doi.org/10.4230/LIPIcs.ICALP.2023.71},
  doi          = {10.4230/LIPIcs.ICALP.2023.71},
  bibsource    = {dblp computer science bibliography, https://dblp.org}
}

@article{ScaledPierFractals,
  author    = {David Furcy and
               Scott M. Summers},
  title     = {Scaled pier fractals do not strictly self-assemble},
  journal   = {Natural Computing},
  volume    = {16},
  number    = {2},
  pages     = {317--338},
  year      = {2017}
}

@inproceedings{Temp1PathsSTOC,
  author    = {Pierre{-}{\'{E}}tienne Meunier and
               Damien Regnault and
               Damien Woods},
  editor    = {Konstantin Makarychev and
               Yury Makarychev and
               Madhur Tulsiani and
               Gautam Kamath and
               Julia Chuzhoy},
  title     = {The program-size complexity of self-assembled paths},
  booktitle = {Proccedings of the 52nd Annual {ACM} {SIGACT} Symposium on Theory
               of Computing, {STOC} 2020, June 22-26, 2020},
  address   = {Chicago, IL, USA},
  pages     = {727--737},
  publisher = {{ACM}},
  year      = {2020},
  url       = {https://doi.org/10.1145/3357713.3384263},
  doi       = {10.1145/3357713.3384263},
  bibsource = {dblp computer science bibliography, https://dblp.org}
}

@inproceedings{DDDIU,
  author    = {Daniel Hader and
               Aaron Koch and
               Matthew J. Patitz and
               Michael Sharp},
  editor    = {Shuchi Chawla},
  title     = {The Impacts of Dimensionality, Diffusion, and Directedness on Intrinsic
               Universality in the abstract Tile Assembly Model},
  booktitle = {Proceedings of the 2020 {ACM-SIAM} Symposium on Discrete Algorithms,
               {SODA} 2020, January 5-8, 2020},
  address   = {Salt Lake City, UT, USA},
  pages     = {2607--2624},
  publisher = {{SIAM}},
  year      = {2020}
}

@inproceedings{TempOneNotIU,
 author = {Meunier, Pierre-\'{E}tienne and Woods, Damien},
 title = {The Non-cooperative Tile Assembly Model is Not Intrinsically Universal or Capable of Bounded Turing Machine Simulation},
 booktitle = {Proceedings of the 49th Annual ACM SIGACT Symposium on Theory of Computing},
 series = {STOC 2017},
 year = {2017},
 isbn = {978-1-4503-4528-6},
 location = {Montreal, Canada},
 pages = {328--341},
 numpages = {14},
 url = {http://doi.acm.org/10.1145/3055399.3055446},
 doi = {10.1145/3055399.3055446},
 acmid = {3055446},
 publisher = {ACM},
 address = {New York, NY, USA},
}

@inproceedings{DirectedNotIU,
  author    = {Jacob Hendricks and
              Matthew J. Patitz and
              Trent A. Rogers},
title = {Universal Simulation of Directed Systems in the abstract Tile Assembly Model Requires Undirectedness},
booktitle = {Proceedings of the 57th Annual IEEE Symposium on Foundations of Computer Science (FOCS 2016), New Brunswick, New Jersey, USA {\rm October 9-11, 2016}},
pages = {800-809},
year={2016}}

@inproceedings{Polygons,
  author    = {Oscar Gilbert and
              Jacob Hendricks and
              Matthew J. Patitz and
              Trent A. Rogers},
title = {Computing in continuous space with self-assembling polygonal tiles},
booktitle = {Proceedings of the Twenty-Seventh Annual ACM-SIAM Symposium on Discrete Algorithms (SODA 2016), Arlington, VA, USA {\rm January 10-12, 2016}},
pages = {937-956},
year = {2016}
}

@inproceedings{Polyominoes,
  author    = {S{\'a}ndor P. Fekete and
              Jacob Hendricks and
              Matthew J. Patitz and
              Trent A. Rogers and
              Robert T. Schweller},
title = {Universal Computation with Arbitrary Polyomino Tiles in Non-Cooperative Self-Assembly},
booktitle = {Proceedings of the Twenty-Sixth Annual ACM-SIAM Symposium on Discrete Algorithms (SODA 2015), San Diego, CA, USA, January 4-6, 2015},
chapter = {12},
pages = {148-167},
year = {2015}
}

@inproceedings{TreeFractals,
  author    = {Kimberly Barth and
              David Furcy and
              Scott M. Summers and
              Paul Totzke},
  title    = {Scaled tree fractals do not strictly self-assemble},
  booktitle = {Unconventional Computation \& Natural Computation (UCNC) 2014, University of Western Ontario, London, Ontario, Canada {\rm July 14-18, 2014}},
  year      = {2014},
  pages = {27--39}
}

@article{jDuples,
  author    = {Jacob Hendricks and
               Matthew J. Patitz and
               Trent A. Rogers and
               Scott M. Summers},
  title     = {The power of duples (in self-assembly): It's not so hip to be square},
  journal   = {Theoretical Computer Science},
  volume    = {743},
  pages     = {148--166},
  year      = {2018}
}

@inproceedings{temp1notIU,
  title     = {Intrinsic universality in tile self-assembly requires cooperation},
  author    = {Pierre-\'Etienne Meunier and
               Matthew J. Patitz and
               Scott M. Summers and
               Guillaume Theyssier and
               Andrew Winslow and
               Damien Woods},
  year = {2014},
  booktitle = {Proceedings of the ACM-SIAM Symposium on Discrete Algorithms (SODA 2014), (Portland, OR, USA, January 5-7, 2014)},
  pages = {752--771},
}

@inproceedings{2HAMIU,
  title = {The two-handed assembly model is not intrinsically universal},
  author = {Erik D. Demaine and Matthew J. Patitz and Trent A. Rogers and Robert T. Schweller and Scott M. Summers and Damien Woods},
  year = {2013},
  booktitle = {40th International Colloquium on Automata, Languages and Programming, ICALP 2013, Riga, Latvia, July 8-12, 2013},
  series = {Lecture Notes in Computer Science},
  publisher = {Springer},
  page={400-412}
}

@article{BryChiDotKarSekTOC,
  author    = {Nathaniel Bryans and
               Ehsan Chiniforooshan and
               David Doty and
               Lila Kari and
               Shinnosuke Seki},
  title     = {The Power of Nondeterminism in Self-Assembly},
  journal   = {Theory of Computing},
  volume    = {9},
  year      = {2013},
  pages     = {1-29},
  ee        = {http://dx.doi.org/10.4086/toc.2013.v009a001},
  bibsource = {DBLP, http://dblp.uni-trier.de}
}

@inproceedings{IUSA,
 author    = {David Doty and
               Jack H. Lutz and
               Matthew J. Patitz and
               Robert T. Schweller and
               Scott M. Summers and
               Damien Woods},
 title     = {The tile assembly model is intrinsically universal},
 booktitle = {Proceedings of the 53rd Annual IEEE Symposium on Foundations of Computer Science},
 series = {FOCS 2012},
 year = {2012},
 location = {New Brunswick, New Jersey},
 pages = {302-310}
}

@inproceedings{GeoTiles,
  author    = {Bin Fu and
               Matthew J. Patitz and
               Robert T. Schweller and
               Robert Sheline},
  editor    = {Artur Czumaj and
               Kurt Mehlhorn and
               Andrew M. Pitts and
               Roger Wattenhofer},
  title     = {Self-assembly with Geometric Tiles},
  booktitle = {Automata, Languages, and Programming - 39th International Colloquium,
               {ICALP} 2012, Warwick, UK, July 9-13, 2012, Proceedings, Part {I}},
  series    = {LNCS},
  volume    = {7391},
  pages     = {714--725},
  publisher = {Springer},
  year      = {2012}
}

@inproceedings{SingleNegative,
  author    = {Matthew J. Patitz and
               Robert T. Schweller and
               Scott M. Summers},
  editor    = {Luca Cardelli and
               William M. Shih},
  title     = {Exact Shapes and {T}uring Universality at Temperature 1 with a Single
               Negative Glue},
  booktitle = {{DNA} Computing and Molecular Programming - 17th International Conference,
               {DNA} 17, Pasadena, CA, USA, September 19-23, 2011. Proceedings},
  series    = {Lecture Notes in Computer Science},
  volume    = {6937},
  pages     = {175--189},
  publisher = {Springer},
  year      = {2011}
}

@inproceedings{CookFuSch11,
  author = "Matthew Cook and Yunhui Fu and Robert T. Schweller",
  title = "Temperature 1 Self-Assembly: Deterministic Assembly in 3{D} and Probabilistic Assembly in 2{D}",
  year = 2011,
  booktitle =  {SODA 2011: Proceedings of the 22nd Annual ACM-SIAM Symposium on Discrete Algorithms},
  publisher =  {SIAM},
}

@inproceedings{SFTSAFT,
  author = "David Doty and Matthew J. Patitz and Dustin Reishus and Robert T. Schweller and Scott M. Summers",
  title = "Strong Fault-Tolerance for Self-Assembly with Fuzzy Temperature",
  year = 2010,
  booktitle = "Proceedings of the 51st Annual IEEE Symposium on Foundations of Computer Science (FOCS 2010)",
  pages = "417--426"
}

@inproceedings{USA,
  author = "David Doty and Jack H. Lutz and Matthew J. Patitz and Scott M. Summers and Damien Woods",
  title = "Intrinsic Universality in Self-Assembly",
 booktitle = "Proceedings of the 27th International Symposium on Theoretical Aspects of Computer Science",
  year = 2009,
pages = "275--286"
}

@article{SolWin07,
  author    = {David Soloveichik and
               Erik Winfree},
  title     = {Complexity of Self-Assembled Shapes},
  journal   = {SIAM Journal on Computing},
  volume    = {36},
  number    = {6},
  year      = {2007},
  pages     = {1544-1569},
  ee        = {http://dx.doi.org/10.1137/S0097539704446712},
  bibsource = {DBLP, http://dblp.uni-trier.de}
}

@phdthesis{Winf98,
  author =	"Erik Winfree",
  title =	"Algorithmic Self-Assembly of {D}{N}{A}",
  school =	"California Institute of Technology",
  year =	"1998",
  month =	"June",
}

@inproceedings{RotWin00,
 author = {Paul W. K. Rothemund and Erik Winfree},
 title = {The Program-size Complexity of Self-Assembled Squares (extended abstract)},
 booktitle = {STOC '00: Proceedings of the thirty-second annual ACM Symposium on Theory of Computing},
 year = {2000},
 pages = {459--468},
 address = {Portland, Oregon, United States},
 publisher = {ACM}
 }

@article{jSSADST,
  author =   "James I. Lathrop and Jack H. Lutz and Scott M. Summers",
  title =    "Strict Self-Assembly of Discrete {S}ierpinski Triangles",
  journal = "Theoretical Computer Science",
  volume = "410",
  year = "2009",
  publisher = {Springer US},
  pages = "384--405"
}

@article{jCCSA,
  author    = {James I. Lathrop and
               Jack H. Lutz and
               Matthew J. Patitz and
               Scott M. Summers},
  title     = {Computability and Complexity in Self-assembly},
  journal   = {Theory Comput. Syst.},
  volume    = {48},
  number    = {3},
  year      = {2011},
  pages     = {617-647},
  ee        = {http://dx.doi.org/10.1007/s00224-010-9252-0},
  bibsource = {DBLP, http://dblp.uni-trier.de}
}

@article{jSADS,
  author    = {Patitz, Matthew J. and Summers, Scott M.},
  title     = {Self-assembly of decidable sets},
  journal   = {Natural Computing},
  volume    = {10},
  number    = {2},
  year      = {2011},
  pages     = {853-877},
  ee        = {http://dx.doi.org/10.1007/s11047-010-9218-9},
  bibsource = {DBLP, http://dblp.uni-trier.de}
}

\newpage
\appendix




\section{Details of the construction of the Sierpinski Triangle (\cref{lem:triangle-strict})}
\label{apdx:Sierpinski}

\triangleStrict*

\begin{figure}
    \begin{subfigure}[t]{.7\textwidth}     
    \centering
    \includegraphics[width=1\linewidth]{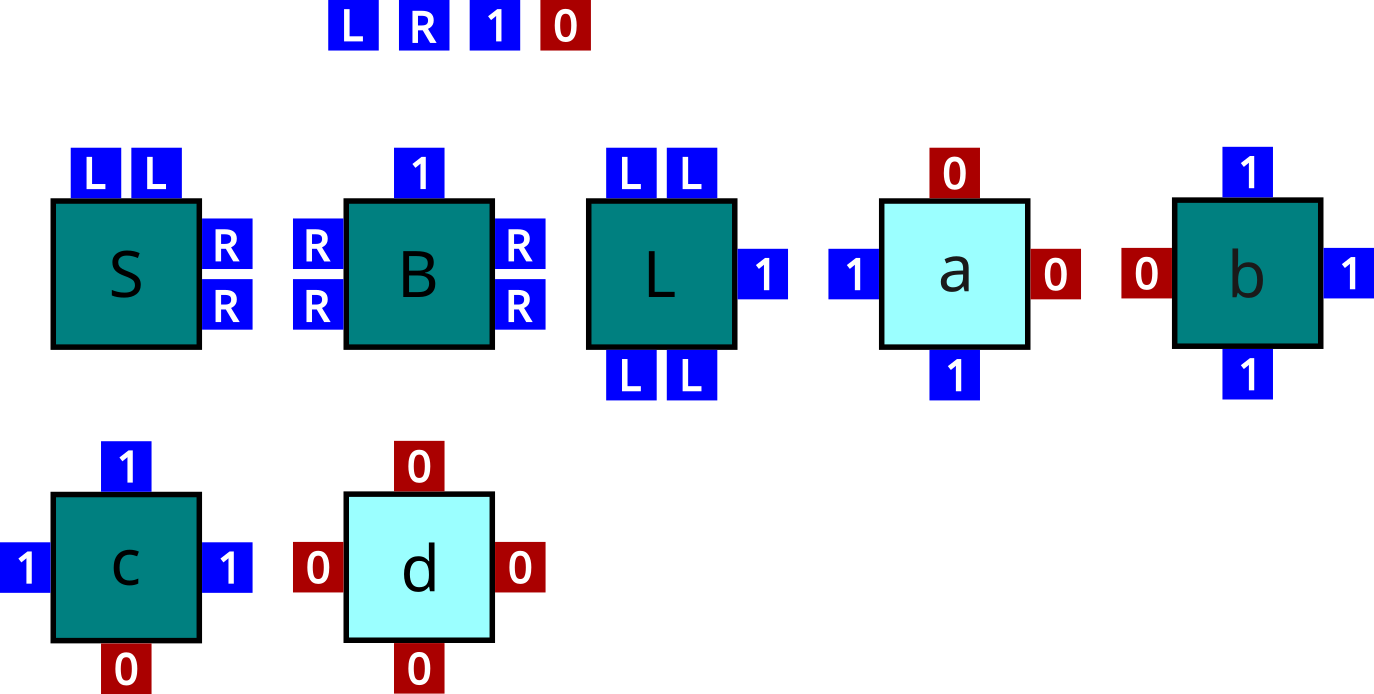}
        \caption{}
        \label{fig:aTAM_weakSTtileSET}
        \end{subfigure}\quad\quad
 
    \begin{subfigure}[b]{.92\textwidth}
        \centering
        \includegraphics[width=1\linewidth]{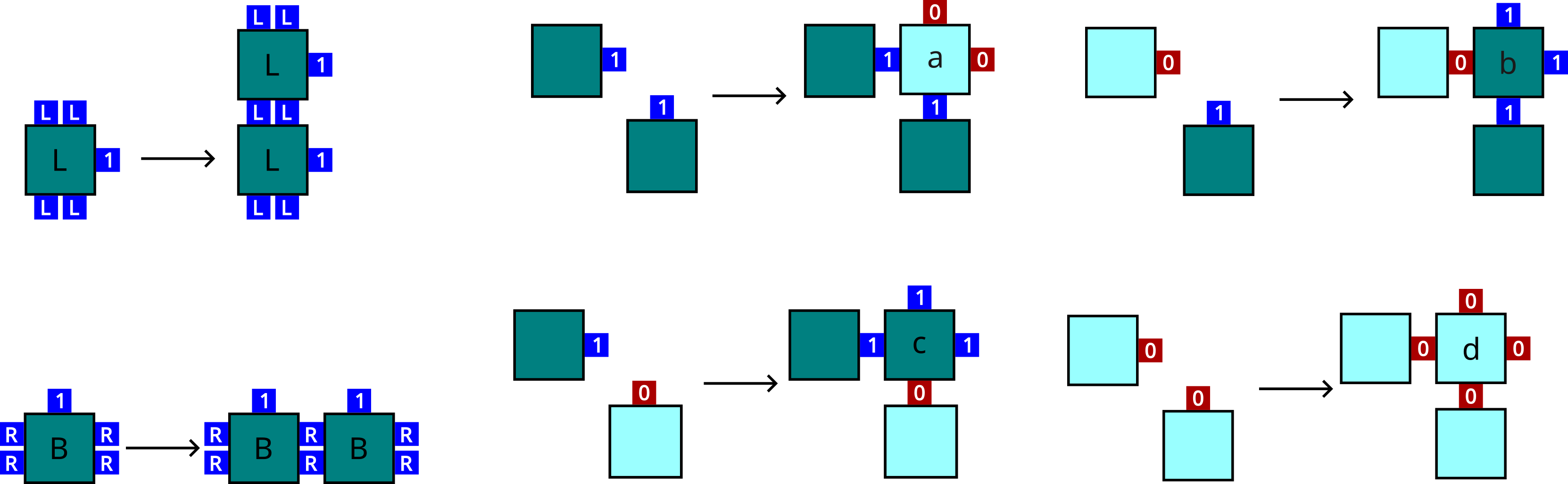}
        \caption{}
        \label{fig:aTAM_weakSTconnections}
    \end{subfigure}\quad\quad

    \begin{subfigure}[b]{.7\textwidth}
        \includegraphics[width=.8\linewidth]{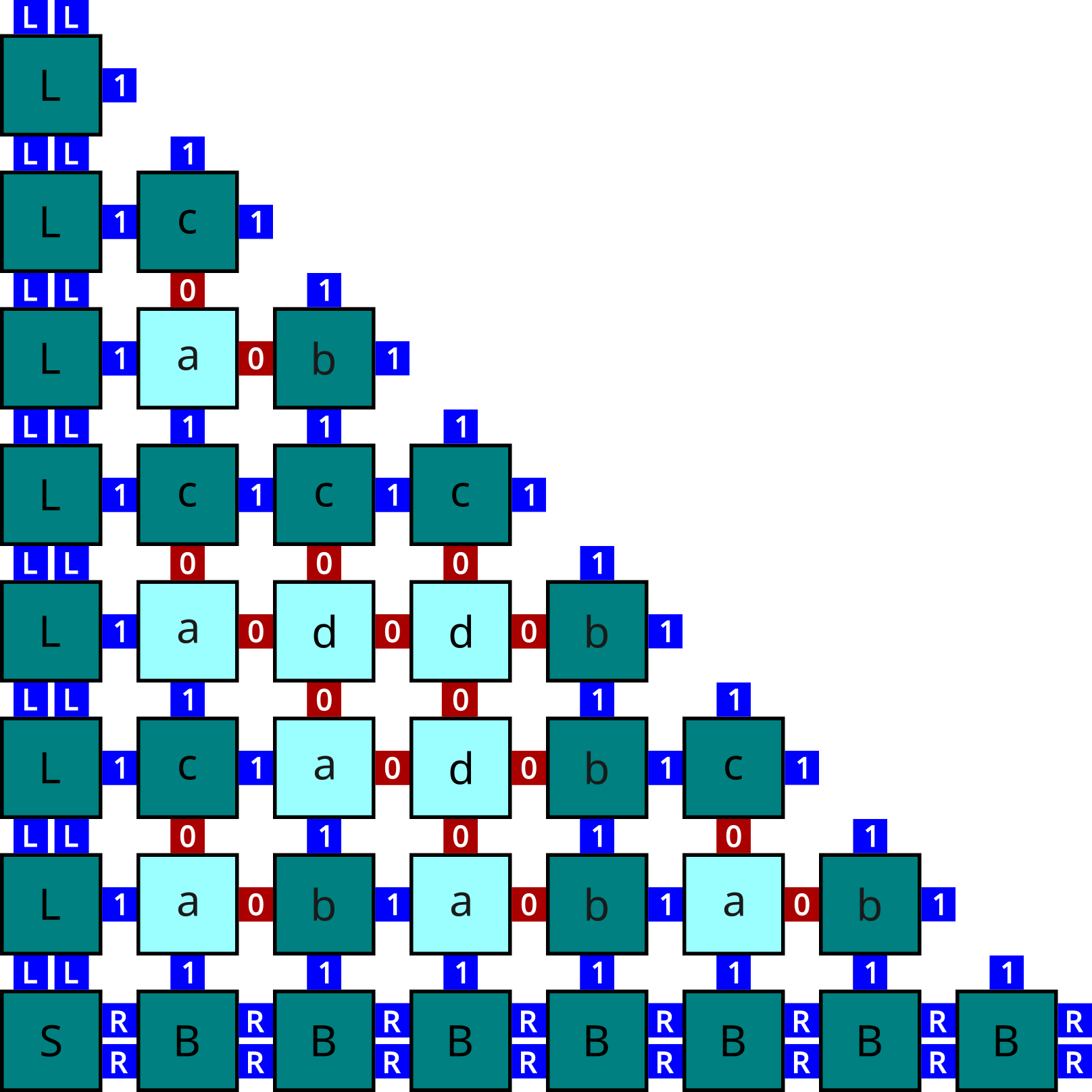}
        \caption{}
        \label{fig:aTAM_weakSTconstructionWglues}
    \end{subfigure}
    \caption{An aTAM system $\calT_{\Delta'}$ that weakly self-assembles the Sierpinski Triangle. This design is standard, and was first shown in \cite{Winf98}. (a) The tile set of an aTAM system which weakly self-assembles the Sierpinski Triangle. The tile labeled ``S'' is the Seed tile, and is placed at $(0,0)$. (b) A diagram showing all possible tile attachments during the assembly of $\calT_{\Delta'}$. For tiles whose glues used to attach may come from multiple tile types in $\calT_{\Delta'}$, an unlabeled placeholder tile type is shown. (c) A snapshot of a producible assembly by $\calT_{\Delta'}$.}\label{fig:aTAM-Sierpinski}
\end{figure}

    \begin{figure}
        \begin{subfigure}[t]{.75\textwidth}
            \centering
            \includegraphics[width=1\linewidth]{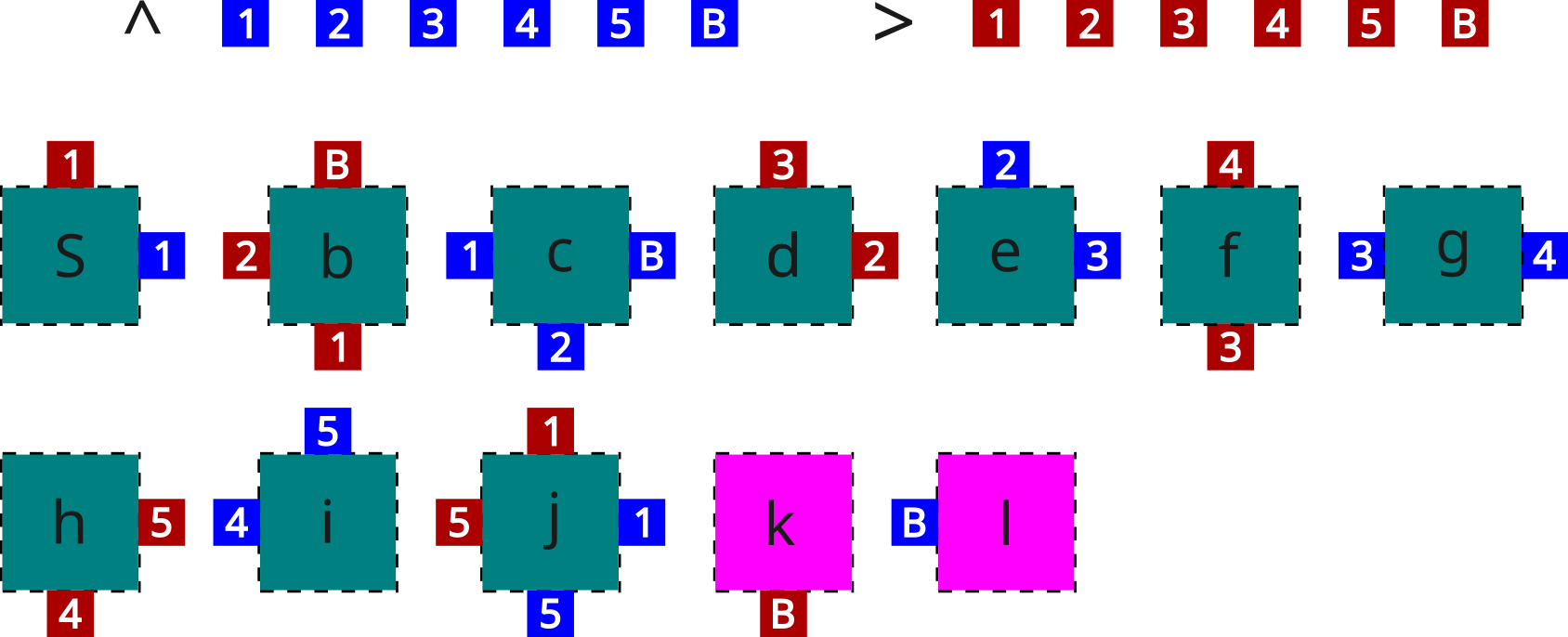}
            \caption{
            }
            \label{fig:SierpinskiTranlge_TileSet}
        \end{subfigure}\quad \quad
        \begin{subfigure}[m]{.75\textwidth}
            \centering
            \includegraphics[width=1\linewidth]{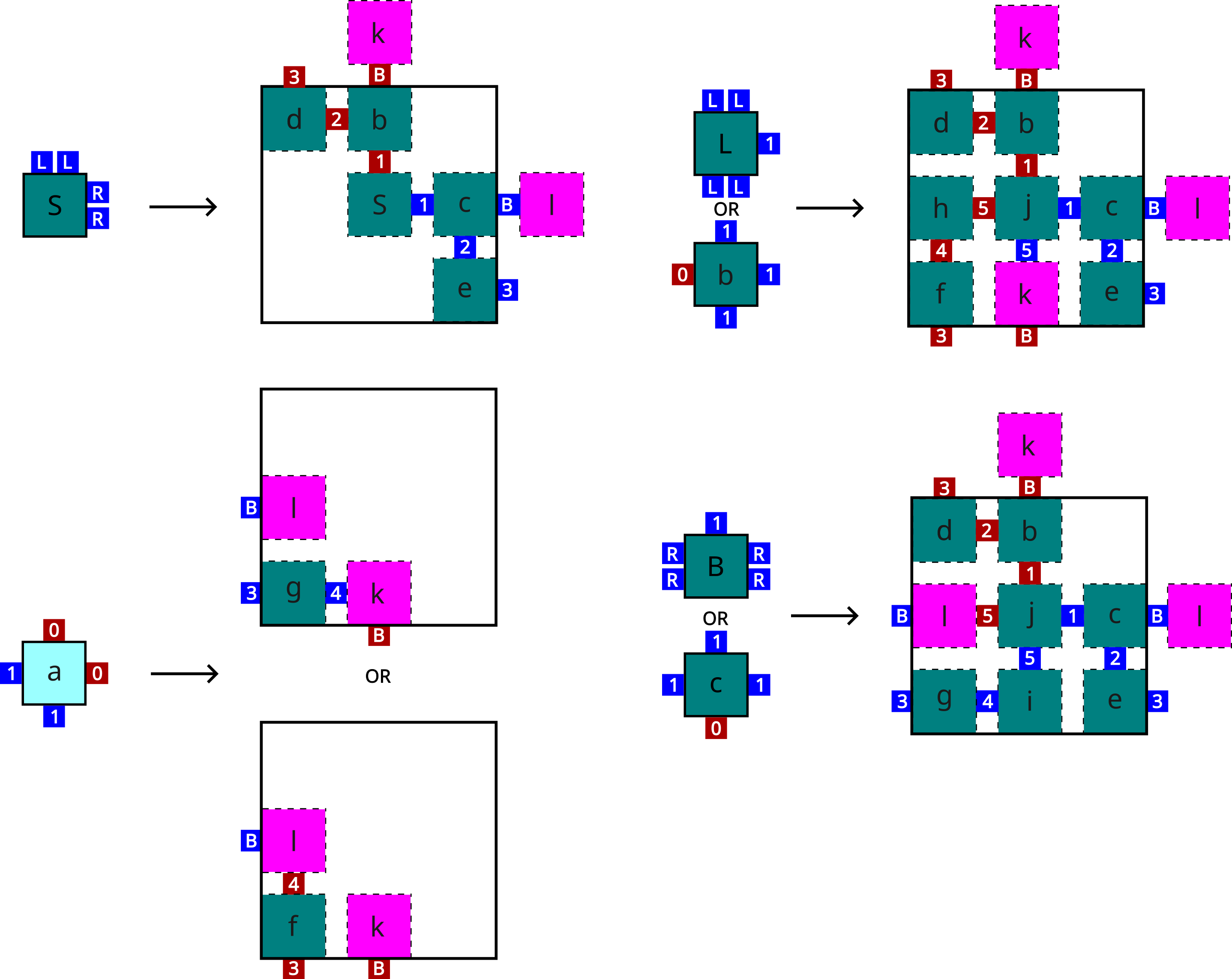}
            \caption{
            }
            \label{fig:SierpinskiTrangle_MacroTiles}
        \end{subfigure}\quad \quad
        \caption{The tile set of a syncTAM system $\calS_\Delta$ that strictly self-assembles our modified Sierpinski triangle by ``simulating'' the aTAM system $\calT_{\Delta'}$. (a) The tile types used in $\calS_\Delta$. (b) A diagram showing the conversion of tile types in $\calT_{\Delta'}$ to macrotiles such that each ``macrotile type'' mimics the behavior of one, or multiple tile types in $\calT_{\Delta'}$. Note that the behavior of `L' and `B' tiles have their behavior encoded in the same macrotile type as tiles `b' and `c', respectively, and that the tile type `d' has no counterpart. }
        \label{fig:Sierpinski_macrotiles}
        \end{figure}
        
    \begin{figure}
        \centering
            \begin{subfigure}[m]{.8\textwidth}
            \centering
            \includegraphics[width=.95\linewidth]{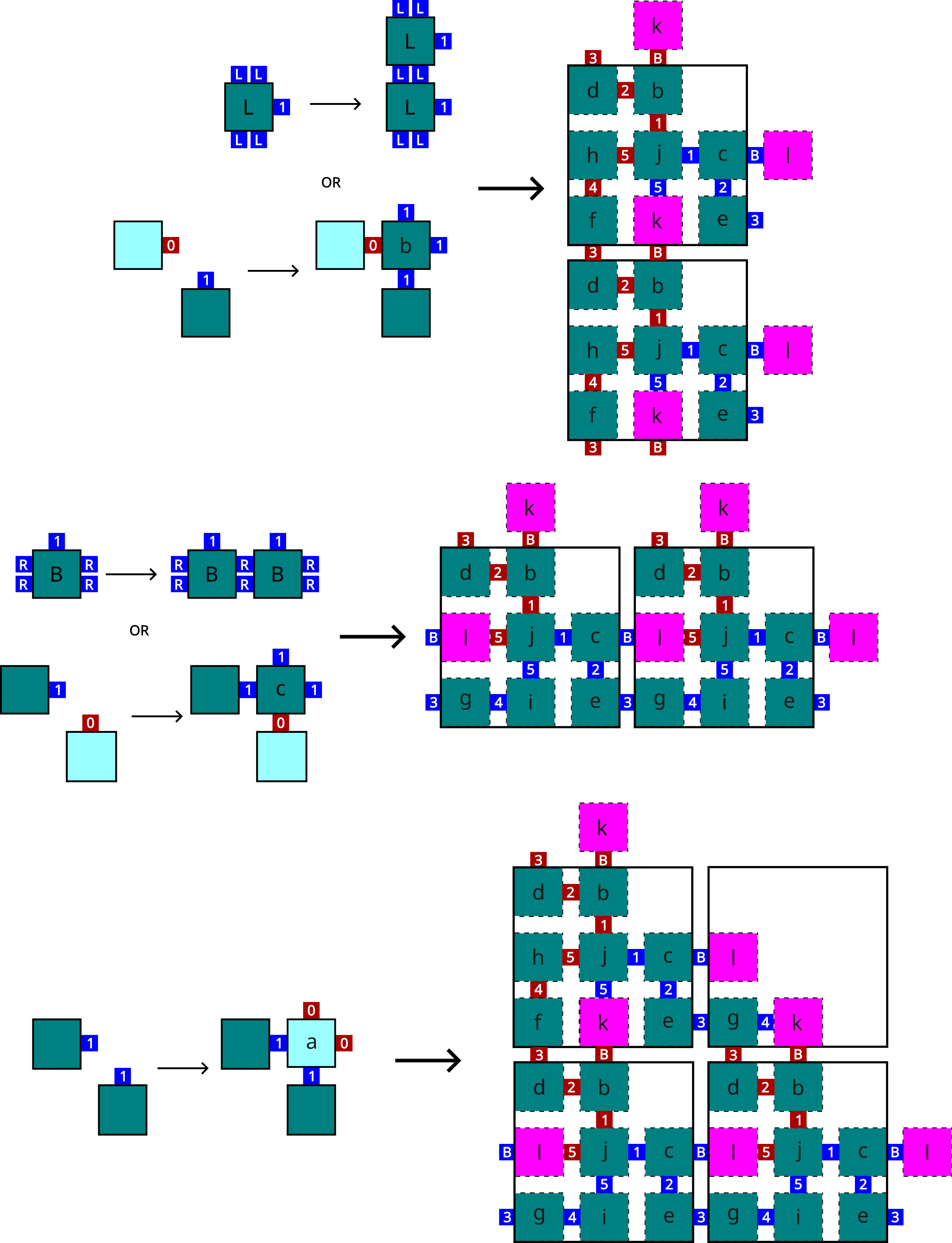}
            \caption{}
            \end{subfigure}

            \begin{subfigure} [m]{.75\textwidth}
            \centering
            \includegraphics[width=.90\linewidth]{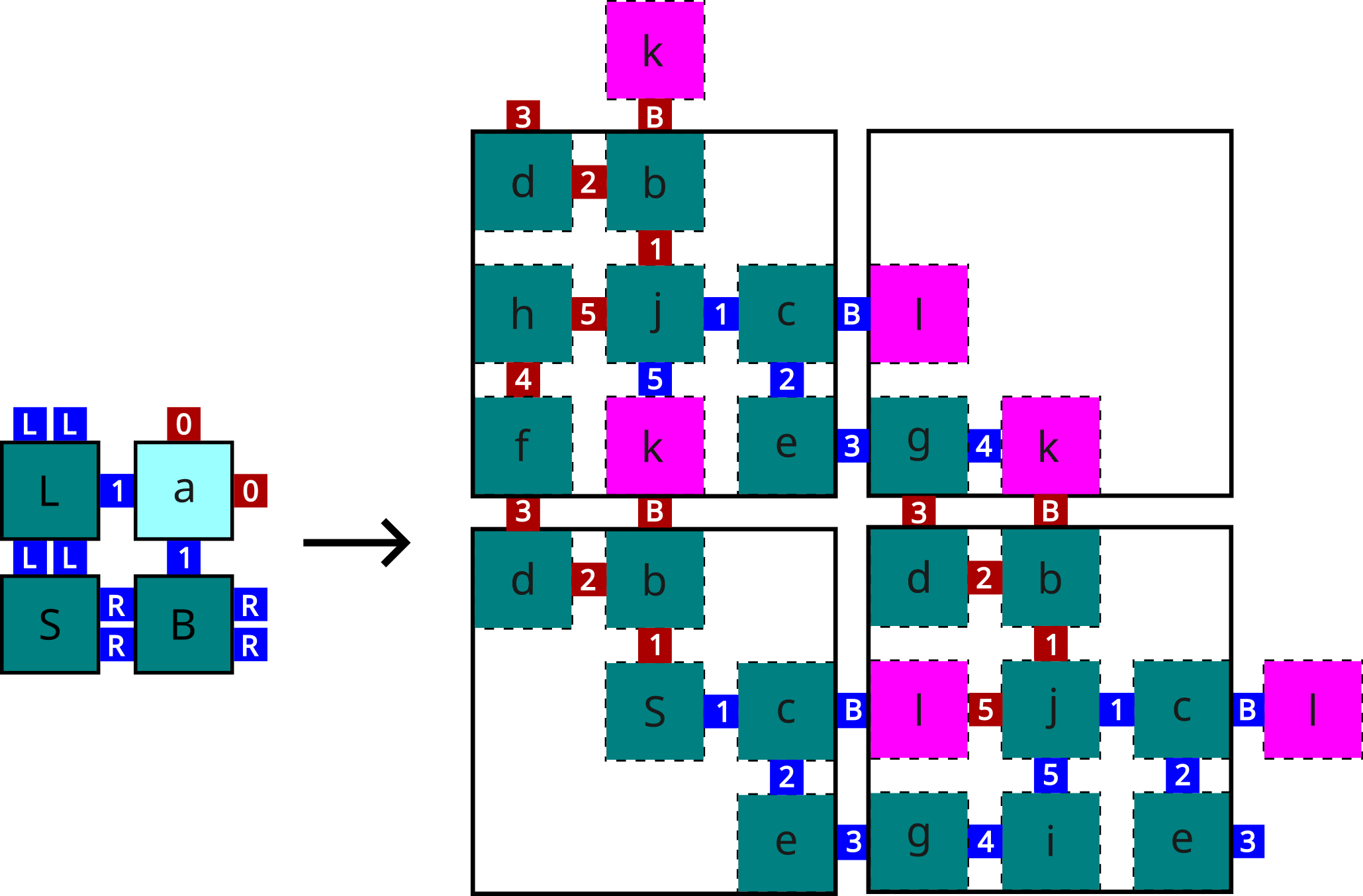}
            \caption{}
                
            \end{subfigure}

        \caption{(a) A diagram which shows tile attachments in the aTAM system $\calT_{\Delta'}$, and their macrotile counterparts in the syncTAM system $\calS_\Delta$. (b) A small assembly producible by $\calS_\Delta$.}
            \label{fig:SierpinskiTrangle_Translations}
    \end{figure}
    \begin{figure}
        \centering
         \begin{subfigure}[m]{\textwidth}
            \centering
            \includegraphics[width=.95\linewidth]{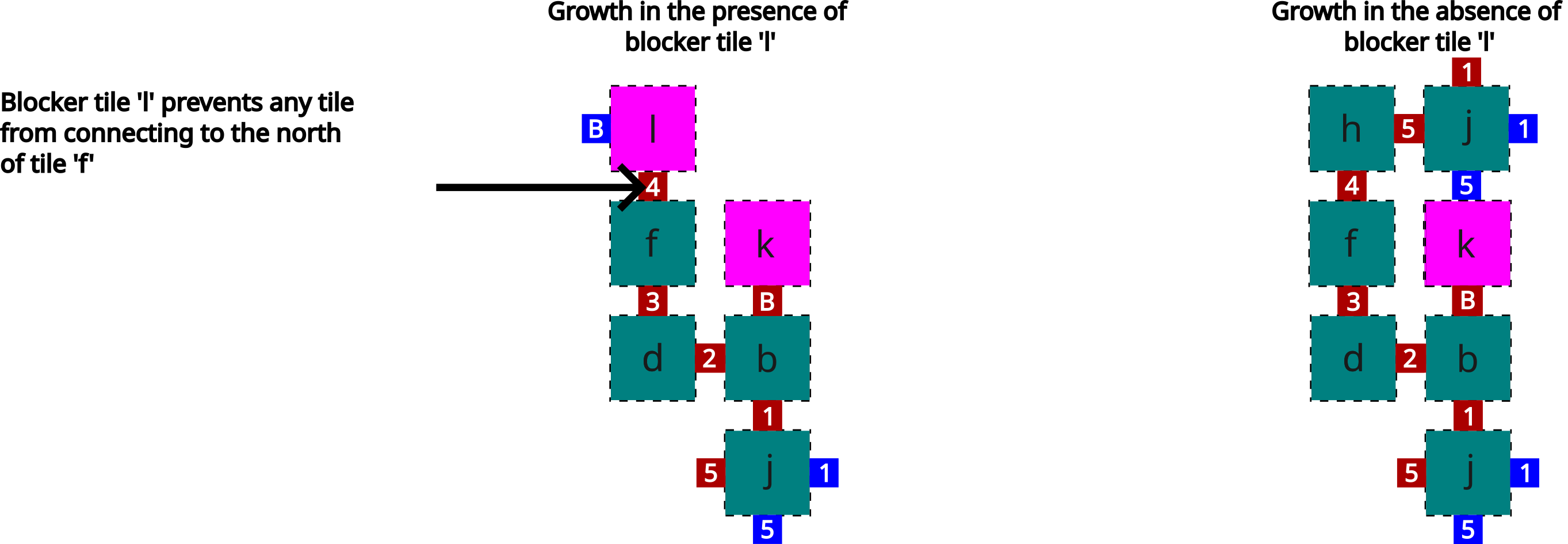}
            \caption{}
            \end{subfigure}

            \begin{subfigure} [m]{.95\textwidth}
            \centering
            \includegraphics[width=.90\linewidth]{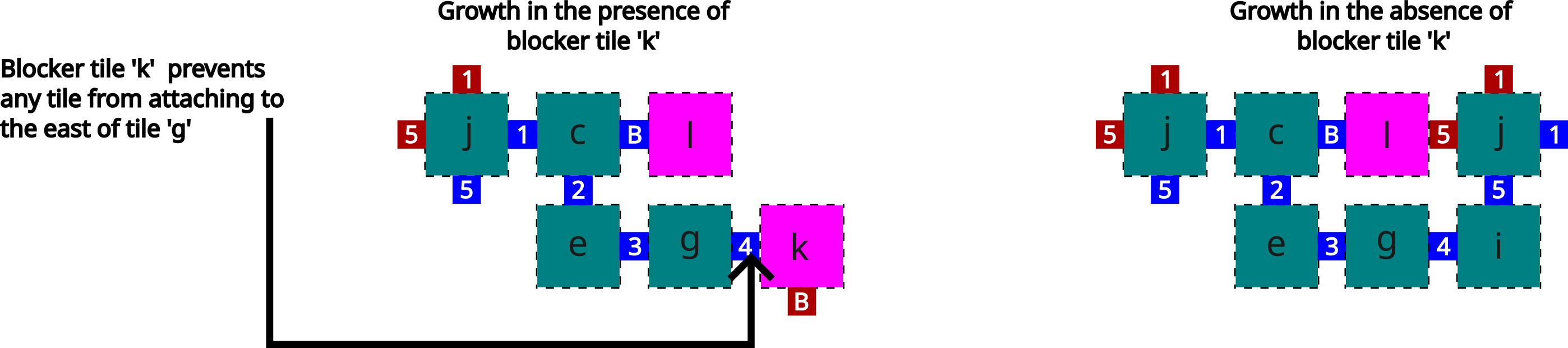}
            \caption{}
                
        \end{subfigure}

        \caption{The functionality of blockers in the SyncTAM system $\calS_\Delta$. (a) The growth of a path of tiles from the north of one of the macrotiles shown in \cref{fig:SierpinskiTrangle_MacroTiles} (excluding the macrotile which simulates the behavior of tile `a'). Note that when the blocker tile `l' is present, the path of tiles forms the macrotile corresponding to tile `a' in $\calT_{\Delta'}$, and when the blocker tile is not present, the path of tiles forms the macrotile corresponding to tile `L' or `b' in $\calT_{\Delta'}$. (b) The growth of a path of tiles from the east of one of the macrotiles shown in \cref{fig:SierpinskiTrangle_MacroTiles} (excluding the macrotile which simulates the behavior of the tile `a'). Note that when the blocker tile `k' is present, the path of tiles forms the macrotile corresponding to tile `a' in $\calT_{\Delta'}$, and when the blocker tile is not present, the path of tiles forms the macrotile corresponding to tile `B' or `c' in $\calT_{\Delta'}$.}
        \label{fig:SierpinskiTriangle_Blockers}
    \end{figure}

\begin{proof}
    The syncTAM system which strictly self-assembles our variant of the Sierpinski Triangle is founded on an aTAM system which \emph{weakly} self-assembles a Sierpinski Triangle. This system is depicted on \cref{fig:aTAM-Sierpinski}. We will call the aTAM system which weakly self-assembles a Sierpinski Triangle $\calT_{\Delta'}$, and we will call the syncTAM system which strictly self-assembles our variant of the Sierpinski Triangle $\calS_\Delta$.

    The full tile set for the syncTAM system can be seen in \cref{fig:Sierpinski_macrotiles}, \cref{fig:SierpinskiTrangle_Translations}, and \cref{fig:SierpinskiTriangle_Blockers} give depictions of how the individual tiles of the temperature-2 aTAM system are simulated by $3 \times 3$ macrotiles formed by tiles of the syncTAM system at temperature-1. The key techniques utilized are: (1) all paths of tiles that can grow through any macrotile are the same length, ensuring that, due to the synchronous nature of tile placements in the syncTAM, all macrotiles that begin growth during the same step also complete growth during the same step, (2) if an aTAM tile would output a $0$ to a neighboring location, the corresponding syncTAM macrotile doesn't initiate any growth into the corresponding neighboring macrotile, and (2) if an aTAM tile would output a $1$ to a neighboring location, the corresponding syncTAM macrotile initiates growth into the corresponding neighboring macrotile as follows. It places a \emph{blocker} tile into the location in the middle of adjacent side of that macrotile, and it also initiates growth of a path into that macrotile by first placing a tile in the southwestern corner. If two neighboring macrotiles output $1$ glues into the same macrotile location, then one of them will nondeterministically get to place its tile in the southwestern corner. However, this will result in there being blockers which prevent a path from continuing from the southwestern corner tile and growth in that macrotile terminates. This handles the case of the input $11$ resulting in output $00$. If, instead, only one neighbor outputs a $1$ into an adjacent macrotile location (via the blocker and the tile in the southwest corner), then it will grow a path which continues such that it outputs $11$, i.e. blockers and paths into the north and east neighbors. This handles the cases for inputs $01$ and $10$. The only other case, for inputs $00$, is trivially handled by the fact that when a $0$ glue is output by a macrotile, there will be no tiles placed in the adjacent macrotile. Therefore, any tile in $\calT_{\Delta'}$ which takes input $00$ will not be present in $\calS_\Delta$. In this way, all possible tile attachments are handled correctly and the pattern shown in \cref{fig:syncTAM-Sierpinski} is generated.

\end{proof}

\newpage

\section{Proof of \cref{thm:directedSim}}
\label{apdx:directedSimProof}

\directedSim*

\begin{proof}
    In the aTAM as well as in the syncTAM, any tile attachment takes place thanks to a set of \emph{predecessors.} The set of directions of those predecessors can be (up to rotation) one of $\{N\}, \{N, E\}, \{N, S\}, \{N, E, S\}, \{N, E, S, W\}$. In this proof, these sets of directions are represented as a binary tree, respectively as follows: $N$, $N \wedge E$, $N \wedge S$, $(N \wedge E) \wedge S$, $(N \wedge E) \wedge (S \wedge E)$. Likewise, for each set $I$ of predecessors, the value of the input glues of the position is a function $i: I \to G$, will be considered as the binary tree obtained by mapping $i$ on the leaves of the representation of $i$. For instance, the input mapping glues $n,e,s$ on the $N,E,$ and $S$, respectively, $\{N \mapsto n, E \mapsto e, S \mapsto s\}: \{N, E, S\} \to G$, is represented as the binary tree $(n \wedge e) \wedge s$.

    Let $\calS = (T, \sigma, \tau)$ be a TAS which is directed in the aTAM, and $G$ its set of glues. Let $g = |G|$ be the number of glues in $T$. The syncTAM system $\calS' = (T', \sigma', 1)$ defined below, with set of glues $G'$, creates a scaled version of $\alpha \in \termasm{S}$, the terminal assembly of $\calS$, 
    in the syncTAM at scale $s = 10g^4 + 1$ with the representation function $R: B_m^{T'} \dashrightarrow T$ also defined below.

    \textbf{Interface between macrotiles and invariants.} The set of glues $G'$ of $\calS'$ is made of two parts: $E \biguplus I$. The set $E$ of \emph{external} glues contains a glue $g'$ (of strength $1$) for each $g \in G$ and $d \in \{N, E, S, W\}$. The set $I$ of \emph{internal} glues will be defined below, as the need for those glues arises in the description of the construction. For each tile $t \in T$, there is a unique tile $\hat{t} \in T'$.

    Given scale factor $s$, the macrotiles of $\calS'$ that map under $R$ to tiles of $\calS$ will be $s \times s$ regions.
    Let $i^* = \lceil s/2 \rceil + 1$; on each side of a macrotile, the $i^*$-th position (counting clockwise) is the \emph{input position} in that direction, with the position immediately before it (counting clockwise) being the \emph{output position}. (These are depicted in \cref{fig:SimZones} by the arrows pointing into and out of each side.)

    
    There is an invariant associated with the representation of the glues which $\calS'$ will uphold:
    \begin{quote}
        for any assembly $\alpha$ of $\calS'$ where some $s \times s$ square $M$ maps to a tile $t$ under $R$,        
        there is an $\alpha'$ such that $\alpha \to_{\calS'} \alpha'$ and the sequence $\vec{g}$ of glues (read clockwise) on each side $d$ of $M$ satisfies either $\vec{g}_{i^*} = t_d$ or $\vec{g}_{s-i^*} = t_d$.
        
    \end{quote}

    That is, every macrotile in $\calS'$ that maps to a tile in $\calS$ will have (or be able to grow into a macrotile that has) on each side a glue, in either the input or output position, that represents the glue on the corresponding side of the tile in $\calS$ that the macrotile maps to.


    Because $\calS$ is directed in the aTAM, for any attachment of a tile of type $t$ at position $z$ in a production $\alpha$ of $\calS$ and any subset of neighbors of $z$ which are filled before $z$ and have glues with binding strength summing to $\tau$ or more facing towards $z$, there must not be any tile other than $t$ that is attachable thanks to this subset of neighbors. 


\paragraph*{Macrotile Anatomy} Each macrotile is divided into regions. The \emph{input position} is the $i^*$-th position of each side of the macrotile (counting clockwise), while the \emph{output position} is the $(s-i^*)$-th. Thus, each input position of a macrotile is aligned with the output position of the neighboring macrotile. For each possible set of inputs $I$, there is a \emph{computation region}, with room for a \emph{computation zone} for each combination of glues $i \in G^I$. All these zones and regions are disjoint and separated by the \emph{wire region}, to be used for communication between them.
As an exception to this rule, the wire region also passes through the computation zones in the $\{N, S\}$ and $\{E, W\}$ computing regions. (These are special cases because they handle the situation when a tile in $\calS$ attaches via the cooperation of two input glues across-the-gap, i.e. from $N$ and $S$ or from $E$ and $W$, as opposed to there being a single input glue or a set of input glues that are all adjacent to each other.)
The \emph{center} of the macrotile is the position at $(\lfloor s / 2 \rfloor, \lfloor s / 2 \rfloor)$; it is outside all regions.
The tile there is always one of the unique $\hat{t}$ tiles that corresponds to a tile of $T$. The representation function is defined by $R(M) = t$ iff the center position of $M$ is occupied by $\hat{t}$.

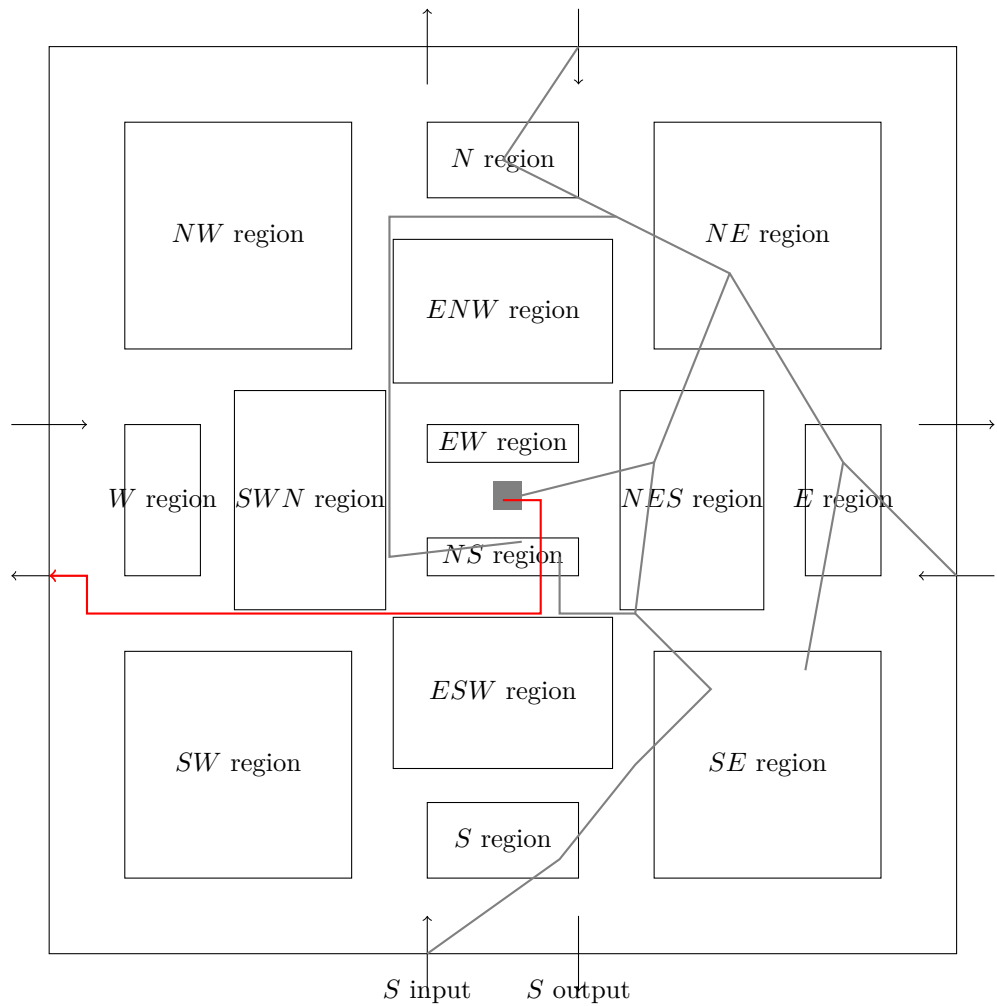
\begin{figure}
    \centering
    \begin{tikzpicture}
    \draw (0,0) rectangle (12, 12);
    \draw[->] (5, -.5) node {$S$ input}-- +(0, 1);
    \draw[->] (7, .5) -- +(0, -1) node {$S$ output};
    \draw[->] (5, 11.5) -- +(0, 1);
    \draw[->] (7, 12.5) -- +(0, -1);
    \draw[->] (-0.5, 7) -- +(1, 0);
    \draw[->] (.5, 5) -- +(-1, 0);
    \draw[->] (11.5, 7) -- +(1, 0);
    \draw[->] (12.5, 5) -- +(-1, 0);

    \fill[gray] (6, 6) + (-.125, -.125) rectangle +(.25, .25);
    \draw (5, 1) rectangle (7, 2) node[pos=.5] {$S$ region};
    \draw (11, 5) rectangle (10, 7) node[pos=.5] {$E$ region};
    \draw (5, 11) rectangle (7, 10) node[pos=.5] {$N$ region};
    \draw (1, 5) rectangle (2, 7) node[pos=.5] {$W$ region};

    \draw (4.55, 2.45) rectangle (7.45, 4.45) node[pos=.5] {$ESW$ region};
    \draw (9.45, 4.55) rectangle (7.55, 7.45) node[pos=.5] {$NES$ region};
    \draw (4.55, 9.45) rectangle (7.45, 7.55) node[pos=.5] {$ENW$ region};
    \draw (2.45, 4.55) rectangle (4.45, 7.45) node[pos=.5] {$SWN$ region};

    \draw (1, 1) rectangle +(3, 3) node[pos=.5] {$SW$ region};
    \draw (1, 8) rectangle +(3, 3) node[pos=.5] {$NW$ region};
    \draw (8, 8) rectangle +(3, 3) node[pos=.5] {$NE$ region};
    \draw (8, 1) rectangle +(3, 3) node[pos=.5] {$SE$ region};

    \draw (5, 5) rectangle +(2, .5) node[pos=.5] {$NS$ region};
    \draw (5, 7) rectangle +(2, -.5) node[pos=.5] {$EW$ region};

    \draw[gray, thick] (7, 12) -- (6, 10.5) -- (9, 9) -- (8, 6.5) -- (6, 6);
    \draw[gray, thick] (12, 5) -- (10.5, 6.5) -- (9,9);
    \draw[gray, thick] (5, 0) -- (6.75, 1.25) -- (7.75, 2.5)-- ++ (1,1) -- (7.75, 4.5) -- (8, 6.5);
    \draw[gray, thick] (7.75, 4.5) -| ++(-1, 0) |- ++(0, .75);
    \draw[gray, thick] (7.5, 9.75) -- ++(-3, 0) -- ++ (0, -4.5) -- (6.25, 5.45);
    \draw[red, thick, ->] (6,6) -| (6.5, 4.5) -| (0.5, 5) -- (0, 5);
    \draw[gray, thick] (10.5, 6.5) -- (10, 3.75);
    
\end{tikzpicture}
    \caption{Anatomy of a macrotile, with the different regions and the flow of information for simulating the attachment of a tile thanks to its North, East and South neighbors. Paths leading to computation zones which will not be activated because they lack an input are omitted. The output (in red) passes through the NS region, which is possible because none of its zones have both inputs. It can do so by splitting and searching for a free path through the region. The output towards the N, E and S output positions will crash into some of the input wires and is not represented.}
    \label{fig:SimZones}
\end{figure}



\paragraph*{Input} 

For each glue $g \in G$ and direction $d$, there is a set of hard-coded tile types such that from glue $g'$ next to the input position in direction $d$, a path grows to the computation zone associated with $(d \mapsto g)$ within the $d$ computation region, where it starts the growth of the computation gadget for $(d \mapsto g)$. 

\paragraph*{Computation gadgets}

A computation gadget for a set of inputs $I$ and the values $i \in G^I$  behaves thus: if there is a unique possible attachment in $\calS$ with inputs $i$, say of tile type $t$, have the macrotile output $t$ (see below for how output proceeds). Otherwise, send $i$ to some computation gadgets having $i$ as a child (within the regions associated with trees having $I$ as a child):
\begin{itemize}
    \item those with two adjacent inputs
    \item those which would result in a unique attachment
\end{itemize}

The computations gadgets for singletons in each direction $d$ are trivial. If glue $g$ has strength at least $\tau$ and there is a unique tile type $t$ with glue $g$ on side $d$, then the gadget is hard-coded to have the macrotile start outputting tile type $t$ (see below). Otherwise, the gadget is hard-coded to forward the input $d \mapsto i$ to the relevant computation gadgets.

For non-singleton sets of input directions, the computation gadget uses a mechanism similar to that of \cref{fig:flagpole} which allows it to wait
for both of its inputs to arrive. When both of its inputs have arrived,
the behavior is the same as in the singleton case: a unique candidate tile attachment $t$ triggers the output of $t$, otherwise the input value is transferred to the parent computation gadgets (if there are any).  

\paragraph*{Output}

When a computation gadget logically outputs the attachment of a tile $t$, it grows a path towards the center of the macrotile and from there to each side, where it passes through the input position before outputting glue $t_d'$ at the output position. The paths from the center to the output positions are designed in such a way that all input and output positions are reached at the same time (i.e. the paths are of the same length, which may require some to ``stall' by growing indirectly toward their destinations for short distances).

\paragraph*{Wires through the NS and EW computation zones}

Wires which need to cross a computation zone within the $\{N, S\}$ or $\{E, W\}$ computation region send two probes, each trying to cross a different input wire of the computation zone then meeting again once they have passed it. The only case where both are blocked is if both inputs are present and they beget a unique tile attachment. In that case, the macrotile is already outputting the value of that attachment (which must be the same value since $\calS$ is directed), so the blocking does not prevent correct output.

The claim is that this system builds a scaled version of $\alpha \in \termasm{S}$. First note that before each output position is reached, all input positions have been blocked. This means that there cannot be a path of attachments between an output position of a macrotile and an input position of the same macrotile.

Moreover, conflicts between paths can only arise if two different computing gadgets have found a unique attachment. In this case, because $\calS$ is directed, they are racing to output the same tile, so the behavior downstream of the center of the macrotile will be the same as if there was no conflict. It is crucial that $\calS$ is directed in the aTAM and not just in the syncTAM, as there might be a race between the macrotile blocking its input positions on output sides and said neighbors outputting a glue there. In $\calS$ running in the syncTAM, such a race is always won by whichever sets of inputs gets there first, so it does not preclude directedness. In the aTAM on the other hand, the attachment can be delayed until both sets inputs are complete, so by directedness, they must all agree on the same unique output. Because there cannot be a dependency path from an output position to an input position of the same macrotile location, any sets of inputs of a macrotile of $\calS'$ which is involved in a conflict between paths must correspond to a competition between attachments which can happen in the aTAM.

If the seed $\sigma$ of $\calS$ consists of a single tile, then the seed $\sigma'$ for $\calS'$ can be created by simply placing the tile $\hat{t}$ that corresponds to the seed tile type $t$ of $\sigma$ into the center location of the macrotile representing $\sigma$. From that, the output paths will grow to the sides of the macrotile, and then (following the argument above) growth of all subsequent macrotiles will correctly match the growth of the tiles of $\calS$ by always growing macrtiles that map, under $R$, to the tile types in the corresponding locations of the terminal assembly $\alpha \in \termasm{S}$. If instead $\sigma$ consists of multiple tiles, then $\sigma'$ simply consists of the corresponding $\hat{t}$ tiles in the centers of the corresponding macrotile locations, all connected by a path of tiles that cross through and block their adjacent input and output positions. Thus, $\calS'$ produces a scaled version of the terminal assembly of $\calS$ for scale factor $s$ and macrotile representation function $R$.


\end{proof}

\section{Proof of \cref{thm:imposs-k-1-sync-shape}}\label{apdx:imposs-k-1-sync-shape}

\noLessSync*

\begin{proof}
    We prove \cref{thm:imposs-k-1-sync-shape} by first defining shape $S_l$ and providing an L-syncTAM system that strictly self-assembles it.
    
    Given some $l > 1$, $S_l$ is an infinite shape consisting of $l - 2$ \texttt{arms} that are single-tile-wide columns that each grow upward infinitely, and a \texttt{cone} that is a ``sideways V'' that grows infinitely to the right, with its sides growing further apart and \texttt{rungs} periodically growing between the two spreading arms of the V. An initial portion of $S_5$ (i.e. $S_l$ for $l = 5$) is shown in \cref{fig:rays-and-cone}.

    The system $\calL_l = (T, \sigma, 2, l)$ strictly self-assembles $S_l$ as follows:
    \begin{enumerate}
        \item From the seed, which is a single tile, a set of $4(l-2)$ tiles grow a row of length $4(l-2)$ to the right. At each position $i$ to the right of the seed where $(i \mod 4) == 0$, the yellow tile at that position has a north facing strength-2 glue that attaches to a pink tile type.

        \item An infinite number of copies of a single pink tile type attach to the north of the yellow tiles at positions $(i \mod 4) == 0$ and then to each other, growing $l - 2$ \texttt{arms} infinitely upward.

        \item To the right of the rightmost yellow tile, a tile of the gold tile type attaches, to initiate growth of the \texttt{cone} (whose initial growth process can be seen in \cref{fig:cone-growth}).

        \item The \texttt{cone} grows in ``periods,'' with each consisting of white tiles that form top and bottom segments that each extend $9$ positions to the right, then grow back to the left $4$ and $5$ positions, respectively.

\begin{figure}
    \centering
    \includegraphics[width=0.85\linewidth]{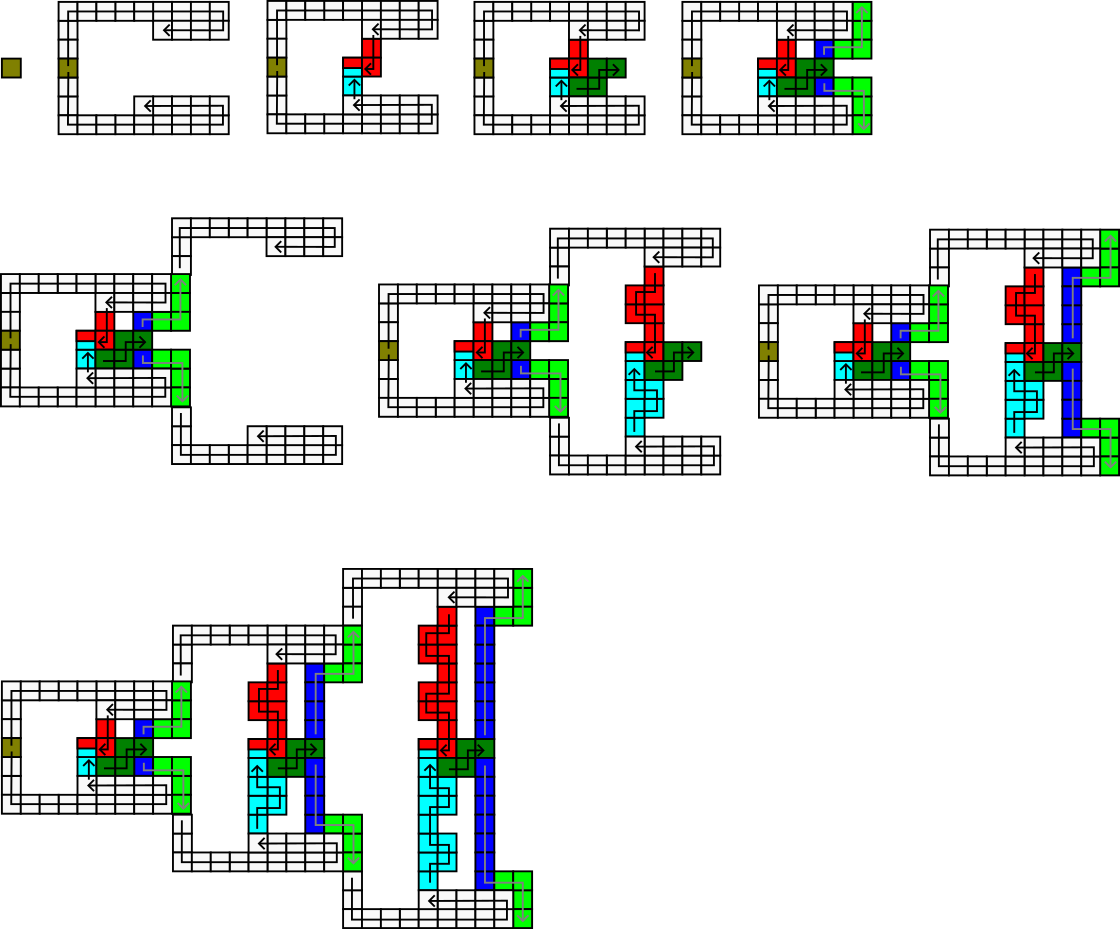}
    \caption{A depiction of the growth sequence of the initial part of the \texttt{cone} potion of the shape $S_5$ from \cref{fig:rays-and-cone}.}
    \label{fig:cone-growth}
\end{figure}

        \item From the final tiles of the top and bottom white segments, red and aqua segments grow downward and upward (respectively). These segments are ``zig-zagging columns'' where the downward growing red columns have ``bumps'' on the left and the upward growing aqua segments have bumps on the right.

        \item The spacing of the white segments ensures that the red and aqua segments collide in such a way that there is one location that is nondeterministically tiled by either a red or aqua tile and that there is guaranteed to be a location where a red and aqua tile can cooperate to allow for the attachment of a green tile.

        \item That cooperation allows for the growth of a row of $4$ green tiles, with the last initiating the upward and downward growth of blue segments.

        \item The upward and downward blue segments each consist of a single tile type that grows a column upward (downward, resp.) until it crashes into a white segment.

        \item When a blue segment crashes into a white segment, there is a location into which a lime green tile can attach via cooperation between a white and blue tile. This initiates growth of 4-tile-long segments of lime green tiles.

        \item The final lime green tiles complete the growth of one period of the \texttt{cone}, with the red, aqua, green, blue, and lime green segments making up the \texttt{rungs}. To the right of the final lime green tiles, the tiles that grow the white segments attach, initiating growth of another period.
    \end{enumerate}

    When growth of $S_l$ begins from the seed, the frontier is initially of size 1. Then, as the yellow row grows and initiates growth of the \texttt{arms}, the size of the frontier increases until, when the yellow portion completes, it is $l-1$ (i.e. one location for each of the $l-2$ arms, and one to the right of the rightmost yellow tile). Then, as the \texttt{cone} grows from the gold tile, it consistently has two frontier locations (one on the top half and one on the bottom half) except during the growth of the 4-tile-long green segments in the middle of each \texttt{rung}, during which there is 1 frontier location in the \texttt{cone}. Thus, once growth of the \texttt{cone} has begun, there are always either $l$ frontier locations or $l-1$ (during the growth of the green segments).
    
    Since there are never more than $l$ frontier locations, the entire frontier can be filled during each assembly step. This ensures that growth of the top and bottom portions of the \texttt{cone} can always remain synchronized, and the lengths of the hard-coded sections of the cone are designed so that, even as the rungs grow taller and taller, the red and aqua segments are always guaranteed to collide in the central row. Because of this, the shape $S_l$ is always formed, with bumps extending to the left of the top halves on the left sides of \texttt{rungs} and to the right of the bottom halves on the left sides of \texttt{rungs}. Therefore, $\calL_l$ strictly self-assembles $S_l$, completing the first portion of the proof.

    We now present the second portion of the proof, showing that for all $l' < l$, there exists no L-syncTAM system $\calL' = (T', \sigma', \tau, l')$ and scale factor $c \in \mathbb{Z}^+$ such that $\calL'$ strictly self-assembles $c \cdot S_l$. We prove this portion by contradiction. Therefore, assume there do exist $l' < l$, $\calL' = (T', \sigma', \tau, l')$, and $c \in \mathbb{Z}^+$ such that $\calL'$ strictly self-assembles $c \cdot S_l$.

   \begin{figure}
    \centering
    \includegraphics[width=1.0\linewidth]{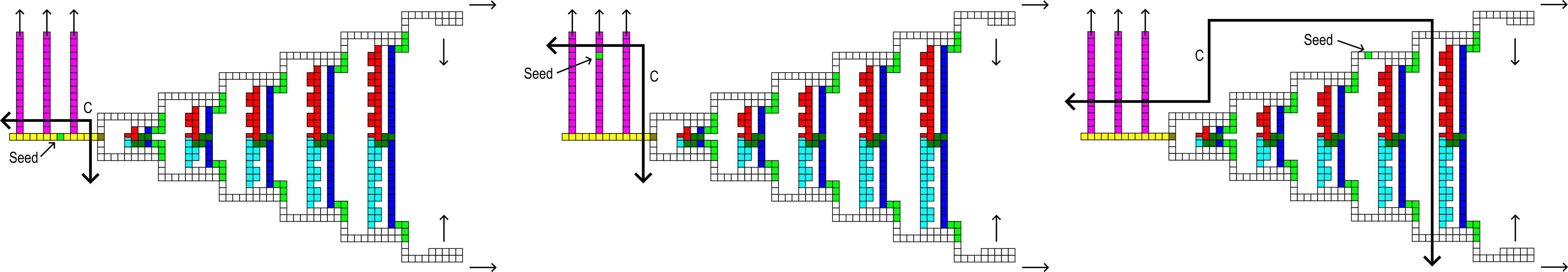}
    \caption{Examples of the shape $S_l$ from the proof of \cref{thm:imposs-k-1-sync-shape} being grown in $\calL'$ with the seed in different locations and possible cuts (labeled as $C$ in each) used to apply \cref{lem:indep}. Note that $\calL'$ is allowed to scale $S_l$ to instead self-assemble $c \cdot S_l$, but that doesn't change the location of the cuts or the fact that they still each meet the criteria of \cref{lem:indep} since the scaling only adds a constant amount of thickness across each portion of the cut.}
    \label{fig:rays-and-cones=cuts}
\end{figure}

    \begin{enumerate}
        \item $\calL'$ is allowed to place its seed anywhere within the shape $c \cdot S_l$. However, since each \texttt{arm} and the \texttt{cone} are infinite, it must be possible to follow an assembly sequence in $\calL'$ that grows portions of each \texttt{arm} and the \texttt{cone} so that it is possible to make a cut that separates the seed-bearing portion of the assembly from infinite segments of the \texttt{arms} and \texttt{cone}. (See \cref{fig:rays-and-cones=cuts} for examples.)
        
        \item We make a cut $C$ through $c \cdot S_l$ that separates it into sub-shapes $s_0, s_1, \ldots, s_{l-1}$, such that $s_0$ consists of the seed and the yellow segment (plus possibly finite portions of the \texttt{arms} and \texttt{cone}), $s_1$ through $s_{l-2}$ each consist of the infinite remainder of one of the \texttt{arms} that isn't contained in $s_0$, and $s_{l-2}$ consists of the infinite portion of the \texttt{cone} that is not contained within $s_0$.
        
        \item We now allow $\calL'$ to proceed following an arbitrary assembly sequence until the number of steps $f$, given by \cref{lem:indep} occur. By \cref{lem:indep} it must be the case that, after some assembly step $f$, each \texttt{arm} must always have at least one frontier location, and the \texttt{cone} must always have at least one, for a total of at least $l-1 \ge l'$ frontier locations during all of the subsequent assembly steps (which are infinite in count).
        
        \item From this point on, we choose an assembly sequence, continuing after step $f$, that proceeds as follows: for every step $f' > f$, a single tile is placed in the \texttt{cone}, and $k' - 1$ tiles are placed, each in a unique \texttt{arm}.

        \item Let $g$ be the number of unique glue types in $T'$, the tile set of $\calL'$. Given the scale factor $c$ at which $\calL'$ is self-assembling $S_l$ to achieve $c \cdot S_l$, we let $h = ((g+1)^{6c}\cdot(6c)!+1)\cdot3+6$. Justification for this number can be found in the proof of Theorem 1.2 of \cite{temp1notIU}, and it is the count of all possible ``window movies'' across cuts of the upward or downward growing segments in the locations where they are one (macro)tile wide. This means that this is the maximum possible number of ways that glues can cross a cut across such a segment.

        \item We let growth in $\calL'$ continue beyond $f'$ until a new period of the \texttt{cone} receives its first tile, and then further until the first time after that during which either a red segment grows downward a distance of $2h$ tiles or an aqua segment grows upward a distance of $2h$. This must occur since the \texttt{rungs} become arbitrarily tall, and for a \texttt{rung} to form, at least one (if not both) of the red and aqua segments must grow at least to the midpoint.

        \item We now rewind the assembly sequence until both the red and aqua segments are shorter than $1.5h$ (noting that one of them already may be shorter than that), and at least one is of length $> h$. Since the \texttt{cone} is growing by the addition of a single tile during each assembly step, we can apply the Window Movie Lemma (WML) of \cite{temp1notIU} and the abstract Tile Assembly Model. Without loss of generality, assume that the red segment is of length $> h$. The WML guarantees that since a window movie must have repeated during its growth to the current length, the segment in between the repeated window movies can be ``pumped'' to create a red segment of arbitrary length (or at least until it crashes into the aqua segment). Since this results in the red segment extending past the midpoint, the shape of the \texttt{rung} is invalid in terms of $c \cdot S_l$ with bumps on the left past the midpoint. (See \cref{fig:cone-pumped} for an example.)

    \begin{figure}
        \centering
        \includegraphics[width=0.6\linewidth]{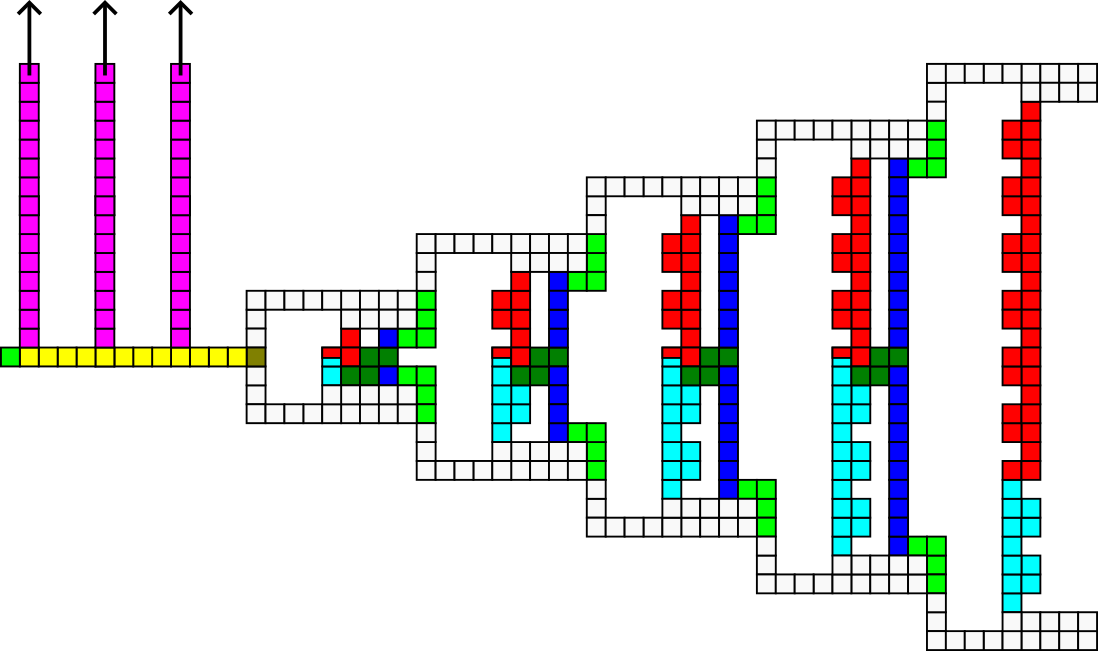}
        \caption{A ``pumped'' version of the red segment of a \texttt{rung} that makes an invalid shape, not matching $c \cdot S_l$, due to the leftward facing bumps of the red segment extending downward into the region in which the bumps should face toward the right.}
        \label{fig:cone-pumped}
    \end{figure}
        
        \item We can create a new, but valid, assembly sequence that continues to place a single tile in each of $k' - 1$ unique \texttt{arms} and a single tile in the red segment of the current period during each step. This results in a \texttt{rung} of the shape described, and therefore the self-assembly of a shape that does not match $c \cdot S_l$, which is a contradiction to the statement that $\calL'$ strictly self-assembles $c \cdot S_l$. Therefore, no system at synchronicity $l' < l$ can strictly self-assemble $S_l$ at any scale factor, and thus \cref{thm:imposs-k-1-sync-shape} is proven.

    \end{enumerate}

\end{proof}

\end{document}